\documentclass[letterpaper,11pt]{article}
\pdfoutput=1
\usepackage{fullpage}

\usepackage[utf8]{inputenc} 
\usepackage[T1]{fontenc}    
\usepackage{hyperref}       
\usepackage{url}            
\usepackage{booktabs}       
\usepackage{amsfonts}       
\usepackage{nicefrac}       
\usepackage{microtype}      

\usepackage{times}
\usepackage{amssymb}
\usepackage{amsmath}
\usepackage{alltt}
\usepackage{appendix}
\usepackage{enumerate}
\usepackage{epsfig}
\usepackage{tabularx}
\usepackage{thmtools}
\usepackage{thm-restate}
\usepackage{mathtools}
\usepackage{color}
\usepackage{url}
\usepackage{comment}
\usepackage{theorem}
\usepackage{mathbbol}
\usepackage{pgfplots}
\usepackage{algorithm}
\usepackage[noend]{algorithmic}
\pgfplotsset{compat=1.13}
\renewcommand*\thesection{\arabic{section}}
\usepackage{xifthen}
\usepackage{graphicx}
\usepackage{caption}
\usepackage{subcaption}

\usepackage{wrapfig}
 
\newenvironment{proof}{\par\noindent\textit{Proof.}}{$\Box$\par\bigskip\par}

\newenvironment{proofof}[1]{\par\noindent\textit{Proof of #1.}}{$\Box$\par\bigskip\par}

\newcommand{\IIF}[2]{\STATE \algorithmicif\ #1\ \algorithmicthen\ #2}
\usepackage{eqparbox}

\usepackage{latexsym}
\newtheorem{theorem}{Theorem}
\newtheorem{lemma}{Lemma}
\newtheorem{corollary}{Corollary}

\newtheorem{problem}{Problem}

\newcommand{\vareps}{\varepsilon}

\newcommand{\cF}{\mathcal{F}}
\newcommand{\hw}{\hat{w}}

\newcommand{\wstar}{w^\star}
\newcommand{\ustar}{u^\star}
\newcommand{\wbest}{w^\text{best}}
\newcommand{\tw}{\tilde{w}}
\newcommand{\hv}{\hat{v}}
\newcommand{\hlambda}{\hat{\lambda}}
\newcommand{\hu}{\hat{u}}
\newcommand{\tu}{\tilde{u}}
\newcommand{\hk}{\hat{k}}
\newcommand{\tk}{\tilde{k}}

\newcommand{\hc}{\hat{c}}
\newcommand{\tc}{\tilde{c}}

\newcommand{\hS}{\hat{S}}

\newcommand{\Sstar}{S^{\star}}
\newcommand{\hx}{\hat{x}}
\newcommand{\tv}{\tilde{v}}

\newcommand{\bbR}{\mathbb{R}}
\newcommand{\R}{\bbR}
\newcommand{\thetahat}{\hat{\theta}}
\newcommand{\thetastar}{\theta^*}
\newcommand{\bbZ}{\mathbb{Z}}

\newcommand{\bbN}{\mathbb{N}_+}
\newcommand{\Otilde}{\widetilde{O}}
\newcommand{\Mmodel}{\mathcal{M}}
\newcommand{\Msupports}{\mathbb{M}}
\newcommand{\dualgreedy}{\text{Dual-Greedy}}
\newcommand{\activeconstraints}{\text{Active-Constraints}}
\newcommand{\optimalvalue}{\text{Opt-Value-of-$\LPdeltadual$}}
\newcommand{\deltarecovery}{\text{$\Delta$-recovery}}
\newcommand{\distributesparsity}{\text{Distribute-Sparsity}}

\newcommand{\recoverprimal}{\text{Recover}}

\newcommand{\deltanosparsity}{\text{$\Delta$-Separated-no-Sparsity}}

\newcommand{\mainalgorithm}{\textsc{LASSP}}

\newcommand{\rtoptimalvalue}[3]{#1 (\log{#2} + \log{#3})}

\newcommand{\mc}{c'}
\newcommand{\mk}{k'}
\newcommand{\mn}{{d'}}

\newcommand{\LPdeltaprimal}{\mathcal{P}}
\newcommand{\LagLPdeltaprimal}[1][]{
    \ifthenelse{\isempty{#1}}{L_{\LPdeltaprimal}}{L_{\LPdeltaprimal(#1)}}
}
\newcommand{\LPdeltaprimalnok}{\mathcal{P}_{\text{no-$k$}}}
\newcommand{\LPdeltadual}{\mathcal{D}}
\newcommand{\dpf}{DP}

\newcommand{\dpfinal}{\overline{DP}}
\newcommand{\boxpack}{\text{Box-Pack}}
\newcommand{\twoddelta}{\text{2D$\Delta$-separated}}
\newcommand{\projlagr}{\textsc{ProjLagr}}

\newcommand{\remark}{\textbf{Remark: }}

\newcommand{\cmax}{c_{max}}
\newcommand{\mcmax}{\mc_{max}}
\newcommand{\allones}{\mathbb{1}}
\newcommand{\allzeros}{\mathbb{0}}

\newcommand{\supp}{\text{supp}}
\newcommand{\sep}{\text{sep}}

\DeclareMathOperator*{\argmin}{arg\,min}
\DeclareMathOperator*{\argmax}{arg\,max}
\DeclareMathOperator*{\val}{VAL}

\newcommand{\prob}[1]{Pr\left[ #1 \right]}

\DeclarePairedDelimiter{\norm}{\lVert}{\rVert}
\DeclarePairedDelimiter{\abs}{\lvert}{\rvert}

\begin{document}

\title{A Fast Algorithm for Separated Sparsity via Perturbed Lagrangians}

\author{
Aleksander Mądry \\ MIT \\ madry@mit.edu
\and Slobodan Mitrović \\ EPFL \\ slobodan.mitrovic@epfl.ch
\and  Ludwig Schmidt \\ MIT \\ ludwigs@mit.edu
}

\date{}
\begin{titlepage}
\maketitle

\begin{abstract}
Sparsity-based methods are widely used in machine learning, statistics, and signal processing. There is now a rich class of structured sparsity approaches that expand the modeling power of the sparsity paradigm and incorporate constraints such as group sparsity, graph sparsity, or hierarchical sparsity. While these sparsity models offer improved sample complexity and better interpretability, the improvements come at a computational cost: it is often challenging to optimize over the (non-convex) constraint sets that capture various sparsity structures. In this paper, we make progress in this direction in the context of separated sparsity -- a fundamental sparsity notion that captures exclusion constraints in linearly ordered data such as time series. While prior algorithms for computing a projection onto this constraint set required quadratic time, we provide a perturbed Lagrangian relaxation approach that computes provably exact projection in only nearly-linear time. Although the sparsity constraint is non-convex, our perturbed Lagrangian approach is still guaranteed to find a globally optimal solution. In experiments, our new algorithms offer a 10$\times$ speed-up already on moderately-size inputs.
\end{abstract}
\end{titlepage}


\section{Introduction}
\label{sec:intro}

Over the past two decades, sparsity has become a widely used tool in several fields including signal processing, statistics, and machine learning.
In many cases, sparsity is the key concept that enables us to capture important structure present in real-world data while making the resulting problem computationally tractable and suitable for mathematical analysis.
Among the many applications of sparsity are sparse linear regression, compressed sensing, sparse PCA, and dictionary learning.

The first wave of sparsity-based techniques focused on the standard notion of sparsity that only constrains the number of non-zeros.
Over time, it became apparent that extending the notion of sparsity to encompass more complex structures present in real-world data can offer significant benefits.
Specifically, utilizing such additional structure often improves the statistical efficiency in estimation problems and the interpretability of the final result.
There is now a large body of work on structured sparsity that has introduced popular models such as group sparsity and hierarchical sparsity \cite{YL06,EM09,BCDH10,MJOB11,HZM11,RRN12,NRWY12,BBC14,HIS15icml}.
These statistical improvements, however, come at a computational cost: the resulting optimization problems are often much harder to solve.
The key reason is that the combinatorial sparsity structures give rise to non-convex constraints.
Consequently, many of the resulting algorithms have significantly worse running time than their ``standard sparsity'' counterparts. This trade-off raises an important question: can we design algorithms for structured sparsity that match the time complexity of commonly used algorithms for standard sparsity?


In this paper, we address this question in the context of the \emph{separated sparsity} model, a popular sparsity model for data with a known minimum distance between large coefficients \cite{hegdeLP,DDJB10,HB11,DSRB13,FMN15}.
In the one-dimensional case, such as time series data, neuronal spike trains are a natural example.
Here, a minimum refractory period ensures separation between consecutive spikes.
In two dimensions, separation constraints arise in the context of astronomical images or super-resolution applications~\cite{HB11,HK15}.

We introduce new algorithms for separated sparsity that run in \emph{nearly-linear} time.
This significantly improves over prior work that required at least quadratic time, which is quickly prohibitive for large data sets.
An important consequence of our fast running time is that it enables methods that utilize separated sparsity yet are essentially as fast as their counterparts based on standard sparsity only.
For instance, when we instantiate our algorithm in compressive sensing, the running time of our method matches that of common methods such as IHT or CoSaMP.

Our algorithms stem from a primal-dual linear programming (LP) perspective on the problem.
Our theoretical findings reveal a rich structure behind the separated sparsity model, which we utilize to obtain efficient methods.
Interestingly, our final algorithm has a very simple form that can be interpreted as a \emph{Lagrangian relaxation} of the sparsity constraint.
In spite of the non-convexity of the constraint, our algorithm is still guaranteed to find the globally optimal solution.

We also show that these algorithmic and theoretical contributions directly translate into empirical efficiency.
Specifically, we demonstrate that, compared to the state of the art procedures, our methods yield an order of magnitude speed-up already on moderate-size inputs.
We run experiments on synthetic data and real world neuronal spike train signals.


\section{Separated sparsity and applications}
In this section, we formally define separated sparsity and the corresponding algorithmic problems.
As a concrete application of our algorithms, we instantiate them in a sparse recovery context that is representative for many statistical problems such as compressed sensing and sparse linear regression.

First, we briefly introduce our notation.
As usual, $[d]$ denotes the set $\{1, \ldots, d\}$.
We say that a vector $\theta \in \bbR^d$ is $k$-sparse if $\theta$ contains at most $k$ non-zero coefficients.
We define the support of $\theta$ as the set of indices corresponding to non-zero coefficients, i.e., $\supp(\theta) = \{ i \in [d] \, | \, \theta_i \neq 0 \}$.
We let $\norm{\theta}$ denote the $\ell_2$-norm of a vector $\theta \in \bbR^d$.

\paragraph{Separated sparsity.}
\emph{Sparsity models} are a natural way to formalize structure beyond ``standard'' sparsity \cite{BCDH10}.
In this work we focus on the the separated sparsity model, defined as follows~\cite{hegdeLP}.
For a support $\Omega \subseteq [d]$, let $\sep(\Omega) = \min_{i \neq j \in \Omega} \abs{i - j}$ be the minimum separation of two indices in the support.
We define the following two sets of supports: set $\Msupports_{\Delta} = \{ \Omega \subseteq [d] \, | \, \sep(\Omega) \ge \Delta \}$, and \emph{$\Delta$-separated sparsity} supports $\Msupports_{k,\Delta} = \{ \Omega \in \Msupports_{\Delta} \, | \, \abs{\Omega} = k\}$.
That is, $\Msupports_{k,\Delta}$ is the set of support patterns containing $k$ non-zeros with at least $\Delta - 1$ zero entries between consecutive non-zeros.

In order to employ separated sparsity in statistical problems, we often want to add constraints based on the support set $\Msupports_{k,\Delta}$ to optimization problems such as empirical risk minimization.
A standard way of incorporating constraints into gradient-based algorithms is via a \emph{projection operator}.
In the context of separated sparsity, this corresponds to the following problem.

\begin{problem}
\label{prob:proj}
For a given input vector $x \in \bbR^d$, our goal is to project $x$ onto the set $\Msupports_{k,\Delta}$, i.e., to find a vector $\hx$ such that
\begin{equation}
\label{eq:proj}
        \hx  \, \in \, \argmin_{x' \in \bbR^d\ :\ \supp\left(x'\right)\in \Msupports_{k,\Delta}} \norm{x - x'} \; .
\end{equation}
\end{problem}
Problem \ref{prob:proj} is the main algorithmic problem we address in this paper.


\paragraph{Sparse recovery.}
Structured sparsity has been employed in a variety of machine learning tasks.
In order to keep the discussion coherent, we present our results in the context of the well-known sparse linear model:
\begin{equation}
\label{eq:sparselinear}
    y \; = \; X \thetastar + e
\end{equation}
where $y \in \R^n$ are the observations/measurements, $X \in \R^{n \times d}$ is the design or measurement matrix, and $e \in \R^n$ is a noise vector.
The goal is to find a good estimate $\thetahat$ of the unknown parameters $\thetastar$ up to the noise level.

The authors of \cite{BCDH10} give an elegant framework for incorporating structured sparsity into the estimation problem outlined above.
They design a general recovery algorithm that relies on a model-specific \emph{projection oracle}.
In the case of separated sparsity, this oracle is required to solve precisely the Problem \ref{prob:proj} stated above.

\section{The algorithm and our results}
\label{sec:contributions}
Given an arbitrary vector $x \in \bbR^d$, Problem \ref{prob:proj} requires us to find a vector $\hx$ such that $\hx \in \Mmodel_{k,\Delta}$ and $\norm{x - \hx}$ is minimized.
We now slightly reformulate the problem.
Let $c \in \bbR^d$ be a vector such that $c_i = x_i^2$ for all $i$.
Then it is not hard to see that this problem is equivalent to finding a set of $k$ entries in the vector $c$ such that each of these entries is separated by at least $\Delta$ and the sum of these entries is maximized.
Hence our main algorithmic problem is to find a set of $k$ entries in an non-negative input vector $c$ so that the entries are $\Delta$-separated  and their sum is maximized.
More formally, our goal is to find a support $\hS$ such that
\begin{equation}
\label{eq:proj-massaged}
        \hS  \, \in \, \argmax_{S \in \Msupports_{k,\Delta}} \sum_{i \in S} c_i \; .
\end{equation}
In the following, we also consider a relaxed version of \eqref{eq:proj-massaged} called $\projlagr$, which is parametrized by a trade-off parameter $\lambda$ and a vector $\tc$:
\[
      \projlagr(\lambda, \tc) \, := \, \argmax_{S \in \Msupports_{\Delta}} \sum_{i \in S} \tc_i + \lambda \left( k - |S|\right) \; .
\]
Intuitively, $\projlagr$ represents a Lagrangian relaxation of the sparsity constraint in Equation~\eqref{eq:proj-massaged}.

\subsection{Algorithm}
Our main contribution is a new algorithm for Problem \ref{prob:proj} that we call \emph{\textbf{L}agrangian \textbf{A}pproach to the \textbf{S}eparated \textbf{S}parsity \textbf{P}roblem} (LASSP).
The pseudo code is given in Algorithm~\ref{alg:main}.

LASSP is a \emph{Las Vegas} algorithm: it always returns a correct answer, but the running time of the algorithm is randomized.
Concretely, LASSP repeats a main loop until a stopping criterion is reached.
Every iteration of LASSP first adds a small perturbation to the coefficients $c$ (see Line \ref{line:main-perturb}).
This perturbation has only a small effect on the solution but improves the ``conditioning'' of the corresponding non-convex Lagrangian relaxation $\projlagr$ so that it returns a globally optimal solution that almost satisfies the constraint.
As we show in Section~\ref{sec:one-iteration-of-main}, we can solve this relaxation (Line~\ref{line:main-lambda}) in nearly-linear time.
After the algorithm has solved the Lagrangian relaxation, it obtains the final support $\hS$ in line~\ref{line:optimal-hc} by solving $\projlagr$ on a slightly shifted $\hlambda$ to ensure that the constraint is satisfied with good probability.

\remark We assume that the bit precision $\gamma$ required to represent the coefficients $c$ is finite, and we provide our results as a function of $\gamma$. For practical purposes, $\gamma$ is usually a constant. Since the solution to Equation \eqref{eq:proj-massaged} is invariant under scaling by a positive integer and $\gamma$ is finite, without loss of generality we assume that $c \in \bbZ^d$.

\begin{algorithm}
	\caption{\mainalgorithm:}
	\label{alg:main}
	{\bfseries Input:} $c \in \bbZ^{d}$, $k \in \bbN$ \\
	{\bfseries Output:} A solution $\hS$ to \eqref{eq:proj-massaged}
	
	\begin{algorithmic}[1]
        \REPEAT
		    \STATE\label{line:main-X} Let $X \in \bbZ^d$ be a vector such that $X_i$ is chosen uniformly at random from $\{0, \ldots, d^3 - 1\}$, $\forall i$
		    \STATE\label{line:main-perturb} Define vector $\tc := d^4 c + X$
		    \STATE\label{line:main-lambda} Let $\hlambda := \argmin_{\lambda \in \bbZ} \projlagr(\lambda, \tc)$. In case of ties, maximize $\hlambda$.
		    \STATE\label{line:optimal-hc} Choose $\hS \in \projlagr\left( \hlambda - \frac{1}{d + 1}, \tc \right)$
		\UNTIL{$\left| \hS \right| = k$}\label{line:main-while}
		\RETURN $\hS$
	\end{algorithmic}
\end{algorithm}

\subsection{Main results}

As our main result, in the following theorem we show that $\mainalgorithm$ runs in nearly linear time and solves Problem~\ref{prob:proj}.
\begin{restatable}{theorem}{theoremrandomized}
\label{thm:introrandomized}
Let $c \in \bbR^d$, and let $\gamma \in \bbN$ be the maximal number of bits needed to store any $c_i$.
There is an implementation of $\mainalgorithm$ that for every $c$ computes a solution $\hc$ to Problem \ref{prob:proj}.
With probability $1 - 1/d$, the algorithm runs in time $O(d (\gamma + \log d))$.
\end{restatable}
Combined with the framework of \cite{BCDH10}, we get the following.
\begin{corollary}
Let $y$, $X$, $\thetastar$, and $e$ be as in the sparse linear model in Equation \eqref{eq:sparselinear}.
We assume that $\supp(\thetastar) \in \Msupports_{k,\Delta}$ and that $X$ satisfies the model-RIP for $\Msupports_{k,\Delta}$.
There is an algorithm that for every $y$ and $e$ returns an estimate $\thetahat$ such that
\[
    \norm{\thetahat - \thetastar}_2 \; \leq \; C \norm{e}_2 \; .
\]
Moreover, the algorithm runs in time $\Otilde(T_X + d)$, where $T_X$ is the time of multiplying the matrices $X$ and $X^T$ by a vector.
\end{corollary}

The corollary shows that, up to logarithmic factors, the running time is dominated by $T_X + d$. This matches the time complexity of standard sparse recovery and shows that we can utilize separated sparsity without a significant increase in time complexity.
Many measurement matrices in compressive sensing enable fast multiplication with $X$ (e.g., a subsampled Fourier matrix), in which case the total running time becomes $\Otilde(d)$.
We validate these theoretical findings in Section~\ref{sec:experiments} by showing that LASSP runs significantly faster than the state of the art algorithm used for sparse recovery with separation constraints, while retaining the same accuracy of the recovered signal.

Our algorithm $\mainalgorithm$ is randomized. However, we also design a deterministic nearly-linear time algorithm and prove the following theorem.
For clarity of exposition, the statement of the deterministic algorithm and the proof of the theorem are deferred to Appendix~\ref{section:deterministic-algorithm}.

\begin{theorem}
\label{thm:introexact}
Let $c \in \R^d$ be the input vector and let $\gamma$ be as in Theorem~\ref{thm:introrandomized}.
Then there is an algorithm that computes a solution $\hc$ satisfying Equation~\eqref{eq:proj} and runs in time $O(d (\gamma + \log{k}) \, \log{d} \, \log \Delta)$.
\end{theorem}

\paragraph{Remark.} The $\Delta$-separated sparsity projection can be reduced to the problem of finding a minimum-weight path of length $k$ on an edge-weighted directed graph. After this work was done, it was pointed to us by Arturs Backurs and Christos Tzamos that, by properly designing the graph, the edge-weights of the corresponding directed graph satisfy the concave Monge property. This further implies that the separated-sparsity problem can be solved in nearly-linear time~\cite{aggarwal1994finding}. However, the algorithm presented in~\cite{aggarwal1994finding} is more complex than the algorithm we provide. The goal of our work was to design a simple, and yet very efficient, algorithm by applying minor but crucial modifications to a widely popular approach, i.e. to Lagrangian relaxation.

\subsection{Further results}

In Appendix~\ref{sec:dp} we present a dynamic programming approach that for a specific, but also natural, family of instances solves the separated sparsity problem in even linear time.

We also consider a natural extension to the 2D-variant of the separated sparsity projection problem and show that it is NP-hard in Appendix~\ref{app:twod}.
Moreover, in Appendix~\ref{app:blocks}, we extend our model to allow for blocks of separated variables and show that our algorithms for $\Msupports_{k,\Delta}$ also applies to the more general variant.
Finally, separated sparsity can be used to model signals in which a longer pattern is repeated multiple times so that any two patterns are at least $\Delta$ apart.
This model is called \emph{disjoint pulse streams}~\cite{HB11}.
Again, the algorithmic core remains the same and algorithms for $\Msupports_{k,\Delta}$ can also be used for this generalization.


\subsection{Additional related work}

The papers \cite{hegdeLP,FMN15} are closely related to our work.
The paper~\cite{hegdeLP} proposed the separated sparsity model, provided a sample complexity upper bound, and gave an LP-based model-projection algorithm. However, they resorted to a black-box approach for solving the LP, that lead to a fairly prohibitive $O(d^{3.5})$ time complexity.
Recently, \cite{FMN15} provided a faster dynamic program for this problem with a time complexity of $O(d^2)$ and also showed a sample complexity lower bound.
The algorithmic aspect of these papers is the main difference from our work: we exploit structure in both the primal and dual formulations of the LP and give an algorithm that {\em provably} runs in \emph{nearly-linear} time.

Beside the papers addressing the core algorithmic question of projecting onto separated sparse vectors, there is much of work utilizing the sparsity model for applications in neural signal processing \cite{DDJB10,HB11,DSRB13} and recovery with coherent dictionaries \cite{DB13,Needham15}.
In the latter application, the separated sparsity constraint enforces that the signal representation only consists of incoherent dictionary atoms.
We expect that our algorithmic techniques will also lead to improvements in the context of these applications.

In addition to separated sparsity, a large body of work on structured sparsity has emerged over the past few years. We refer the reader to the surveys \cite{DE11,BJMO12b,Wain14,HIS15beatcs} for an overview. The line of work most relevant to our paper is the \emph{model-based compressive sensing} framework introduced in \cite{BCDH10}, which is also the starting point for \cite{hegdeLP,FMN15}. While the framework provides a general recovery scheme based on a model-projection oracle satisfying Equation \eqref{eq:proj}, it does not provide any guidance on how to design such oracles for a specific sparsity model. We address precisely this problem for separated sparsity with our nearly-linear time algorithm.

Recently, two papers have proposed a fairly general framework for deriving model-projection oracles via \emph{graph sparsity}, i.e., sparsity structures that can be defined through connected components in a graph on the signal coefficients \cite{HZM11,HIS15icml}.
This framework generalizes several previously studied sparsity models such as block sparsity, tree sparsity, and cluster sparsity.
Moreover, the paper \cite{HIS15icml} gives projections that run in nearly-linear time.
Interestingly, these tools do not apply to the separated sparsity model we study in our paper.
Intuitively, graph sparsity captures structures in which the non-zero coefficients are \emph{clustered together}, while the separated sparsity model achieves a reduction in sample complexity for the opposite reason: the non-zero coefficients are \emph{far apart}.
Moreover, the algorithms in \cite{HIS15icml} are \emph{approximate} and project into a sparsity model with a relaxed sparsity constraint.
As we explain in Appendix~\ref{sec:noapprox}, the separated sparsity model requires more careful control over the output sparsity in order to achieve a meaningful sample complexity improvement over ``standard'' sparsity.
We circumvent this issue by providing an \emph{exact} projection onto the separated sparsity model.
\section{Proof of correctness and the roadmap}
\label{sec:structural-results}

We begin our analysis by proving that $\mainalgorithm$ returns a correct result \emph{if the algorithm terminates}.
As we will see later, establishing termination is the crucial part of the analysis.
Nevertheless, the following lemma is a useful warm-up for understanding how the different pieces of our algorithm fit together.
\begin{lemma}\label{lemma:correctness}
    When $\mainalgorithm$ terminates, it outputs a support $\hS$ such that $x$ restricted to $\hS$ is a solution to Problem~\ref{prob:proj}.
\end{lemma}
\begin{proof}
    As we have argued above, the problems in Equations~\eqref{eq:proj} and ~\eqref{eq:proj-massaged} are equivalent.
    So, we show that $\mainalgorithm$ outputs a solution to the problem in Equation~\eqref{eq:proj-massaged}.
    
    Let $\hS$ be the set returned by $\mainalgorithm$. By the condition of the loop in Line~\ref{line:main-while}, we have $|\hS| = k$. So, as $\hS \in \Msupports_{\Delta}$ (see the definition of $\projlagr$), we have $\hS \in \Msupports_{k, \Delta}$.
    
    Now, towards a contradiction, assume that support $\hS$ is not a solution to the problem in Equation~\eqref{eq:proj-massaged}, while support $\Sstar$ is. This implies that $\sum_{i \in \Sstar} c_i > \sum_{i \in \hS} c_i$. Now since, without loss of generality, we assumed that $c \in \bbZ^d$, the last inequality implies $\sum_{i \in \Sstar} c_i \ge 1 + \sum_{i \in \hS} c_i$, and hence
    \begin{equation}\label{eq:cstar-and-tc}
        \sum_{i \in \Sstar} d^4 c_i \ge d^4 + \sum_{i \in \hS} d^4 c_i \; .
    \end{equation}
    Observe that for any support $S \in \Msupports_{k, \Delta}$, the term $\lambda \left(k - |S| \right)$ equals zero, and recall that $\hS, \Sstar \in \Msupports_{k, \Delta}$. Furthermore, by the definition of the random vector $X$ in line~\ref{line:main-X} and from \eqref{eq:cstar-and-tc}
    \begin{eqnarray*}
        \sum_{i \in \Sstar} \left(d^4 c_i + X_i \right)& \ge & \sum_{i \in \Sstar} d^4 c_i \ge d^4 + \sum_{i \in \hS} d^4 c_i \\
        & > & \sum_{i \in \hS} \left( d^4 c_i + X_i \right).
    \end{eqnarray*}
    Since $\Sstar \in \Msupports_{\Delta}$, this chain of inequalities contradicts Line~\ref{line:optimal-hc} of $\mainalgorithm$ which chooses $\hS$ as an optimal solution to $\projlagr\left(\hlambda - 1/(d + 1), d^4 c + X \right)$. This further implies that $\hS$ is a solution to Problem \ref{prob:proj}.
\end{proof}

\subsection{Roadmap}
Lemma~\ref{lemma:correctness} shows that $\mainalgorithm$ outputs the right answer if it terminates.
But does $\mainalgorithm$ terminate on every input?
Answering this question is the most intricate part of this paper.
We split the proof in two main pieces.
The first part is Section~\ref{sec:duality-and-back}, where we provide an alternative view on the separated sparsity problem based on linear programming duality.
The duality view paves the way towards proving our main results.
In particular, we show that the subroutines in $\mainalgorithm$ can be implemented quickly.
\begin{lemma}\label{lemma:single-iteration-of-main}
    Single iteration of $\mainalgorithm$ can be implemented to run in time $O(d (\gamma + \log{d}))$.
\end{lemma}
The second part is Section~\ref{sec:active-constraints}, where we further study the duality view on separated sparsity.
We show that after perturbing the input instance in Line~\ref{line:main-perturb} of $\mainalgorithm$, the support obtained with a shifted $\hlambda$ in Line~\ref{line:optimal-hc} has cardinality $k$ with high probability.
\begin{restatable}{lemma}{lemmasingleiteration}
\label{lemma:a-single-iteration}
    Algorithm $\mainalgorithm$ runs only a single iteration with probability at least $1 - 1/d$.
\end{restatable}
Together with Lemma~\ref{lemma:correctness}, these results yield Theorem~\ref{thm:introrandomized}.
\vspace{-5pt}
\section{Part I -- To duality and further}
\label{sec:duality-and-back}
\vspace{-5pt}
We now analyze the running time of a single iteration of $\mainalgorithm$.
We provide a series of equivalences, as illustrated in Figure~\ref{fig:reduction-diagram}, in order to exploit structure in the separated sparsity problem.
More precisely, we start with a linear programming (LP) view on separated sparsity.
It has already been shown that this viewpoint yields a totally unimodular LP \cite{hegdeLP}, which implies that the LP has an integral solution.
Hence solving the LP solves the separated sparsity projection in Problem \ref{prob:proj}.
However, prior work did not utilize this connection to reason about the power of the Lagrangian relaxation approach to the problem.

We begin our detailed analysis of the LP with the dual program $\LPdeltadual$.
By strong duality, the value of $\LPdeltadual$ equals the value of the primal LP.
Then we cast $\LPdeltadual$ as minimization of LP $\LPdeltadual_{\lambda}$ over $\lambda$.
This reduction will play the central role in our analysis and connect $\LPdeltadual_{\lambda}$ to Line~\ref{line:main-lambda} of $\mainalgorithm$.
\vspace{-5pt}
\subsection{The LP Perspective}
\label{sec:dual}
\vspace{-5pt}
We start with a linear programming view on problem~\eqref{eq:proj-massaged} by considering its LP relaxation denoted by $\LPdeltaprimal$:
\begin{alignat*}{4}
	\text{maximize}   &\qquad & c^T u \\
	\text{subject to} &\qquad & \sum_{i = 1}^d{u_i} & = k & \qquad & \\
					&\qquad & \sum_{j = i}^{\min\{i + \Delta - 1, n\}}{u_j} & \le 1 & \qquad & \forall i = 1 \ldots d \\
					&\qquad & u_i  & \ge 0 &\qquad & \forall i = 1 \ldots d
\end{alignat*}
Given an LP $\mathcal{A}$ we use $\val{\mathcal{A}}$ to denote its optimal objective value. As already noted, $\LPdeltaprimal$ is totally unimodular and thus there always exists an optimal solution to it that is integral \cite{nemhauser1988integer}. For completeness, we provide the proof of total unimodularity in Appendix~\ref{app:dual}.
\\
\remark This implies that a solution to $\LPdeltaprimal$ can be used to obtain a solution to the separated sparsity model projection: if $\ustar$ is an optimal solution to $\LPdeltaprimal$, we can derive the optimal support of \eqref{eq:proj-massaged} from the non-zero entries among the LP variables $\ustar_1, \ldots, \ustar_d$.
It is unclear, however, if there is a way to directly solve this LP fast, e.g. it is not known how to solve $\LPdeltaprimal$ directly in time matching the running time of our algorithm $\mainalgorithm$.
\begin{wrapfigure}[14]{R}{0.5\textwidth}
    \begin{center}
        \centerline{\includegraphics[width=0.45\textwidth]{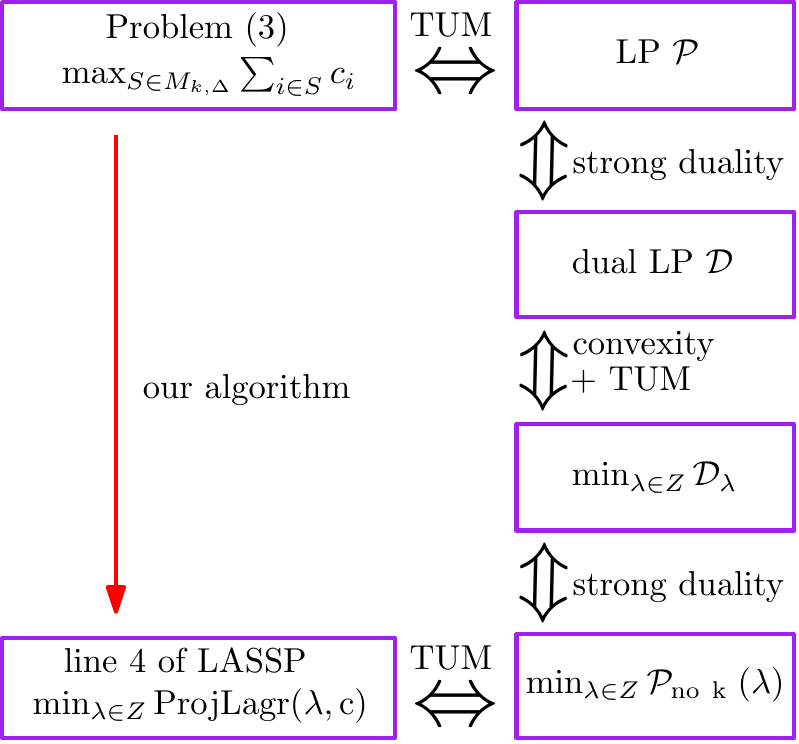}}
        \caption{The values of the problems in the diagram are equal. Every equivalence relation carries structural information that we utilize in our analysis.
        }
        \label{fig:reduction-diagram}
    \end{center}
\end{wrapfigure}

A key step in our approach is understanding the separated sparsity structure from the dual point of view.
The dual LP to $\LPdeltaprimal$, denoted by $\LPdeltadual$, is given as follows
\begin{alignat*}{4}
	\text{minimize}   &\qquad & w_0 k + \sum_{i = 1}^{d} w_i \\
	\text{subject to} &\qquad & w_0 + \sum_{\substack{j\ :\ j \ge 1 \text{ and } \\ j \le i \le j + \Delta - 1}}{w_j} & \ge c_i & \qquad & \forall i = 1 \ldots d \\
					&\qquad & w_i  & \ge 0 &\qquad & \forall i = 1 \ldots d \\
					&\qquad & w_0  & \in \mathbb{R} &\qquad &
\end{alignat*}
Then, as $\LPdeltaprimal$ is integral, so is $\LPdeltadual$.
\begin{restatable}{corollary}{integraloptimal}
\label{corollary:y0-integral-optimal}
    For an integer $k$ and a vector of integers $c$, there exists $\hw$ such that $\hw$ is an optimum of $\LPdeltadual$ and $\hw_0 \in \bbZ$.
\end{restatable}
We also define $\LPdeltadual_{\lambda}$ as the LP $\LPdeltadual$ in which the variable $w_0$ is set to $\lambda$. Now, it is not hard to show the following lemma, whose proof is deferred to Lemma~\ref{lemma:LPdeltadual-is-convex} in Appendix~\ref{app:dual}.
\begin{lemma}\label{lemma:dual-convex-in-lambda}
    $\LPdeltadual_{\lambda}$ is convex with respect to $\lambda$.
\end{lemma}
From the definition of $\LPdeltadual$, the following equality holds $\val{\LPdeltadual} = \min_{\lambda \in \bbR} \val{\LPdeltadual_{\lambda}}$. Furthermore, Corollary~\ref{corollary:y0-integral-optimal} implies that it is sufficient to consider $\lambda$ in $\bbZ$ only, i.e.
\begin{equation}\label{eq:dual-from-dual-lambda}
    \val{\LPdeltadual} = \min_{\lambda \in \bbZ} \val{\LPdeltadual_{\lambda}}.
\end{equation}
Now, Lemma~\ref{lemma:dual-convex-in-lambda} implies that we can obtain $\val{\LPdeltadual}$ by applying ternary search over $\lambda$ on function $\val{\LPdeltadual_{\lambda}}$.

\subsection{Implementing one iteration of $\mainalgorithm$ efficiently}
\label{sec:one-iteration-of-main}
We now derive the final connection between $\LPdeltadual$ and $\mainalgorithm$ which will enable us to obtain $\hlambda$ at line~\ref{line:main-lambda} in nearly-linear time. To that end, consider $\LPdeltaprimalnok(\lambda)$ defined as
\begin{alignat*}{4}
	\text{maximize}   &\qquad & c^T u + \lambda (k - \allones^T u) \\
	\text{subject to} &\qquad & \sum_{j = i}^{\min\{i + \Delta - 1, d\}}{u_j} & \le 1 & \qquad & \forall i = 1 \ldots d \\
					&\qquad & u_i  & \ge 0 &\qquad & \forall i = 1 \ldots d
\end{alignat*}
Observe that compared to $\LPdeltaprimal$, $\LPdeltaprimalnok(\lambda)$ does not contain the sparsity constraint. Furthermore, the LP $\LPdeltaprimalnok(\lambda)$ is a relaxed version of $\projlagr\left(\lambda, c \right)$.
Also, as $\LPdeltaprimal$ is, then $\LPdeltaprimalnok(\lambda)$ is totally unimodular. Now this sequence of conclusions results in the following.
\begin{corollary}\label{corollary:primal-nok-projlagr}
    Problems $\LPdeltaprimalnok(\lambda)$ and $\projlagr\left(\lambda, c \right)$ are equivalent.
\end{corollary}
To obtain the final connection, we consider $\LagLPdeltaprimal$ given by
\[
    \LagLPdeltaprimal := \min_{\lambda \in \bbR} \LPdeltaprimalnok(\lambda).
\]
Next, note that the dual of $\LPdeltaprimalnok(\lambda)$ is $\LPdeltadual_{\lambda}$. Hence, the strong duality implies $\val{\LPdeltadual_{\lambda}} = \val{\LPdeltaprimalnok(\lambda)}$. That together with $\val{\LPdeltadual} = \min_{\lambda \in \bbR} \val{\LPdeltadual_{\lambda}}$ yields that $\LPdeltadual$ and $\LagLPdeltaprimal$ coincide as functions in $\lambda$. Furthermore, since we can solve $\LPdeltadual$ by applying ternary search over integral values of $\lambda$ and $\LPdeltadual_{\lambda}$, we can solve $\LagLPdeltaprimal$ by applying ternary search over integral values of $\lambda$ and function $\LPdeltaprimalnok(\lambda)$. But since $\LPdeltaprimalnok(\lambda)$ and $\projlagr\left(\lambda, c \right)$ are equivalent, we can also obtain $\hlambda$ at line~\ref{line:main-lambda} of $\mainalgorithm$ by applying ternary search over $\lambda$.

Now, it is very easy to see that for an optimal solution $\wstar$ of $\LPdeltadual$ we have $\wstar_0 \le \max_i \left|c_i\right|$. It is also not hard to show that there is an optimal solution such that $\wstar_0 \ge -(k - 1) \max_i \left|c_i\right|$ (see Lemma~\ref{lemma:ystar-lower-bound}). Therefore, in order to find optimal $\hlambda$ it suffices to execute $O\left(\log{\max_i \left|c_i\right|} + \log{k}\right) = O(\gamma + \log{d})$ iterations of ternary search.

Every iteration of the ternary search invokes $\projlagr$, which can be implemented to run in linear time.
\begin{lemma}\label{lemma:solve-projlagr}
    Given $\lambda$ and $\hc \in \bbR^d$, there is an algorithm that finds support $\hS \in \projlagr\left(\lambda, \hc \right)$ in time $O(d)$.
\end{lemma}
\begin{proof}
    Observe that for a fixed $\lambda$, solving $\projlagr\left(\lambda, \hc\right)$ is equivalent to finding support $S' \in \Msupports_{\Delta}$ that maximize $\sum_{i \in S'} \left(\hc_i - \lambda\right)$. Hence, we can reinterpret $\projlagr\left(\lambda, \hc\right)$ as follows: given a vector $\tc := \hc - \lambda \allones$, select a subset of $[d]$ of indices (not necessarily $k$ of them)  so that (i) every two indices are at least $\Delta$ apart, and (ii) the sum of the values of $\tc$ at the selected indices is maximized.
    
    This task can be solved by standard dynamic programming in the following way. For every $i$, we define $s_i$ to be the maximum value of the described task restricted to the first $i$ indices of $\tc$. Then, it is easy to see that $s_{i + 1} = \max\{s_i, s_{i + 1 - \Delta} + \tc_{i + 1}\}$. Namely, we can either decide not to select index $i + 1$, in which case the best is already contained in $s_i$; or, we can decide to select index $i + 1$ which has value $\tc_{i + 1}$ and for the rest we consider $s_{i + 1 - \Delta}$. Therefore, $s_d$ can be obtained in time $O(d)$. Now it is easy to reconstruct the corresponding support in linear time. For completeness, we provide a full algorithm and a detailed proof in Lemma~\ref{lemma:full-proof-proflagr}, Appendix~\ref{app:proj-lagr}.
\end{proof}
Putting all together proves Lemma~\ref{lemma:single-iteration-of-main}.
\vspace{-5pt}
\section{Part II -- Active constraints}
\label{sec:active-constraints}

Note that the chain of equivalences present in Figure~\ref{fig:reduction-diagram} shows that $\min_{\lambda} \projlagr(\lambda, c)$ outputs the sum of coordinates of an optimal solution of problem~\ref{eq:proj}. Prior to our work, it was not even known how to obtain this value in time faster than $O(d k)$, while our result shows we can compute it in nearly-linear time. So, it is natural to ask whether the same relaxation also outputs a support of size $k$? The answer is, unfortunately, no. To see that, consider the example: $c = (4, 7, 5, 0, 0, 5, 8, 5)$, $\Delta = 2$, and $k = 3$. For $\lambda > 2$ every solution $\ustar_{\lambda}$ to $\projlagr(\lambda, c)$ is such that $\ustar_{\lambda}$ has less than $k$ non-zeros. On the other hand, for every $\lambda \le 2$ the exists a solution $\ustar_{\lambda}$ such that $\ustar_{\lambda}$ contains more than $k$ non-zeros. Therefore, there is no $\lambda$, neither $\lambda - 1/(d + 1)$, for which $\projlagr(\lambda, c)$ provably outputs a support of cardinality $k$. This also suggests that the perturbation we apply in lines~\ref{line:main-X}-\ref{line:main-perturb} is essential!

Instead of studying lines~\ref{line:main-X}-\ref{line:optimal-hc} of $\mainalgorithm$ directly, we shift our focus to $\LPdeltadual_{\lambda}$. In particular, we exhibit very close connection between its structure and the sparsity of the primal solution, which we present via the notion of "active constraints".
Then we use these findings in our analysis to show that slight perturbation of the input instance, while not affecting the value of the solution, makes it possible to obtain a solution to problem~\ref{eq:proj-massaged} by applying Lagrangian relaxation.

\subsection{Solving $\LPdeltadual_{\lambda}$}
Observe that once we fixed the value of $w_0$, all remaining constraints in $\LPdeltadual_{\lambda}$ are ``local'' since they only affect a known interval of length $\Delta$. They are also ordered in a natural way. 
As a result, we can solve $\LPdeltadual_{\lambda}$ by making a single pass over these variables.
Starting with $w_1$ and all variables set to $0$, we consider each constraint from left to right and increase the variables to satisfy these constraints in a lazy manner. That is, if in our pass we reach a constraint with index $i$ that is still not satisfied, we increase the value of $w_i$ until that constraint becomes satisfied and then move to the next constraint. Given $c$ and $\lambda$, algorithm $\dualgreedy$, i.e. Algorithm~\ref{alg:dual}, formalizes this approach whose analysis appears in Appendix~\ref{app:dual}.
\begin{algorithm}
	\caption{\dualgreedy:}
	\label{alg:dual}
	{\bfseries Input:} $c \in \bbZ^d$, $\lambda \in \bbR$ \\
	{\bfseries Output:} an optimal solution $w$ to $\LPdeltadual$ such that $w_0 = \lambda$
	
	\begin{algorithmic}[1]
		\STATE $w \gets \allzeros$; $\quad w_0 \gets \lambda$ \COMMENT{initialize the output}
		\STATE $\mathit{sum}_{\Delta} \gets 0$
		\COMMENT{store the sum of the most recent $\Delta$ variables of $w$}
		\FOR{$i := 1 \ldots d$} 
			\IIF{$i - \Delta \ge 1$}{$\mathit{sum}_{\Delta} \gets \mathit{sum}_{\Delta} - w_{i - \Delta}$}
			\STATE $\mathit{diff} \gets c_i - (w_0 + \mathit{sum}_{\Delta})$
			\COMMENT{compute "how far" constraint $i$ is from being tight}
			\IIF{$\mathit{diff} > 0$}{$w_i \gets \mathit{diff}$} \COMMENT{if constraint $i$ is not satisfied, make it tight}
			\STATE $\mathit{sum}_{\Delta} \gets \mathit{sum}_{\Delta} + w_i$
			\COMMENT{update the sum of the $\Delta$ most recent variables of $w$}
		\ENDFOR	
		\RETURN $w$ \COMMENT{solution to $\LPdeltadual$ s.t.  $w_0 = \lambda$}
	\end{algorithmic}
\end{algorithm}

\subsection{Tracking the change of $\LPdeltadual_{\lambda}$}
\label{subsec:active-constraints}
\begin{figure*}[t!]
    \begin{center}
        \begin{subfigure}[b]{0.47 \textwidth}
            \centerline{\includegraphics[width=\columnwidth]{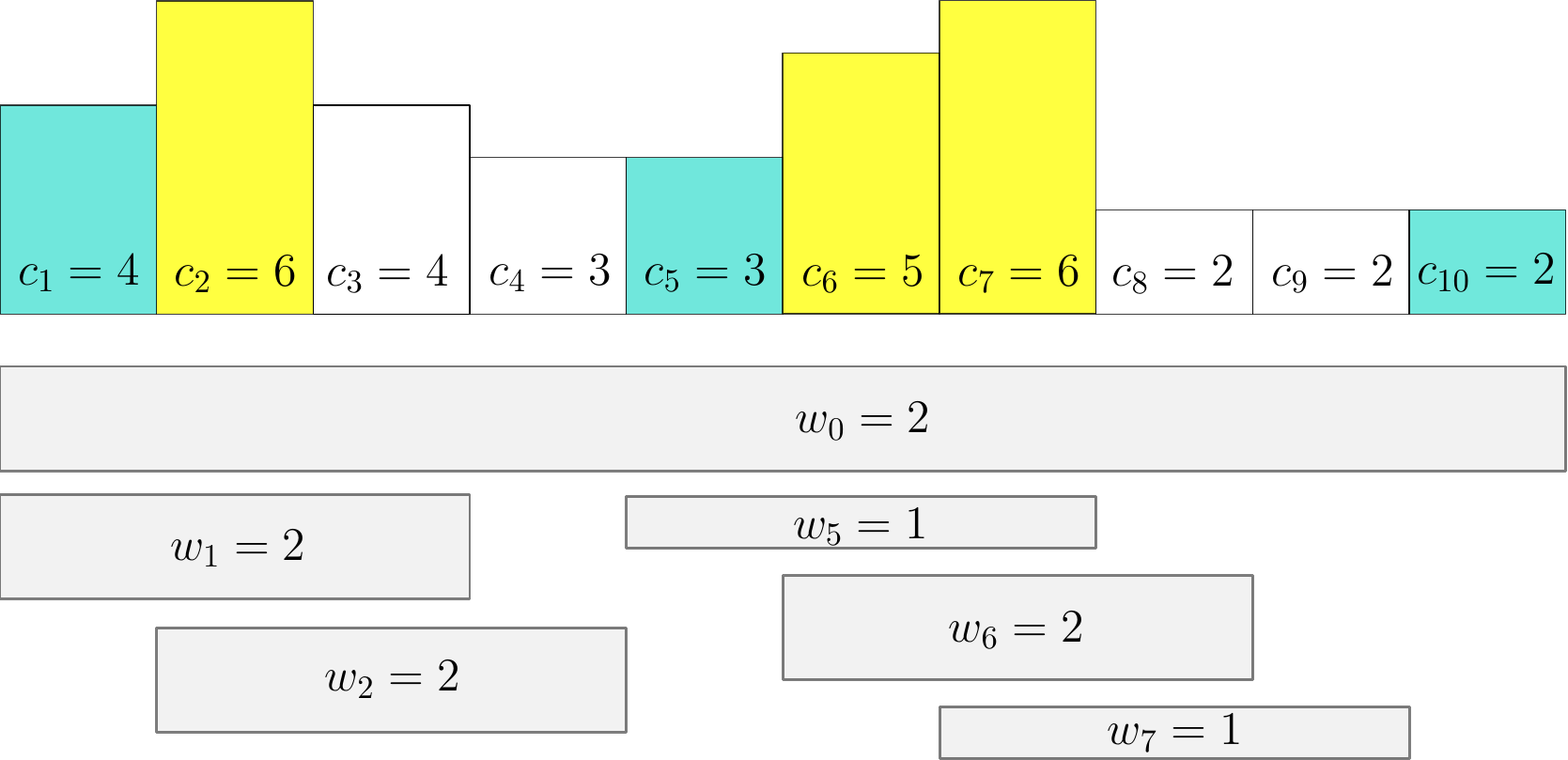}}
        \end{subfigure}
        \hfill 
        \begin{subfigure}[b]{0.47 \textwidth}
            \centerline{\includegraphics[width=\columnwidth]{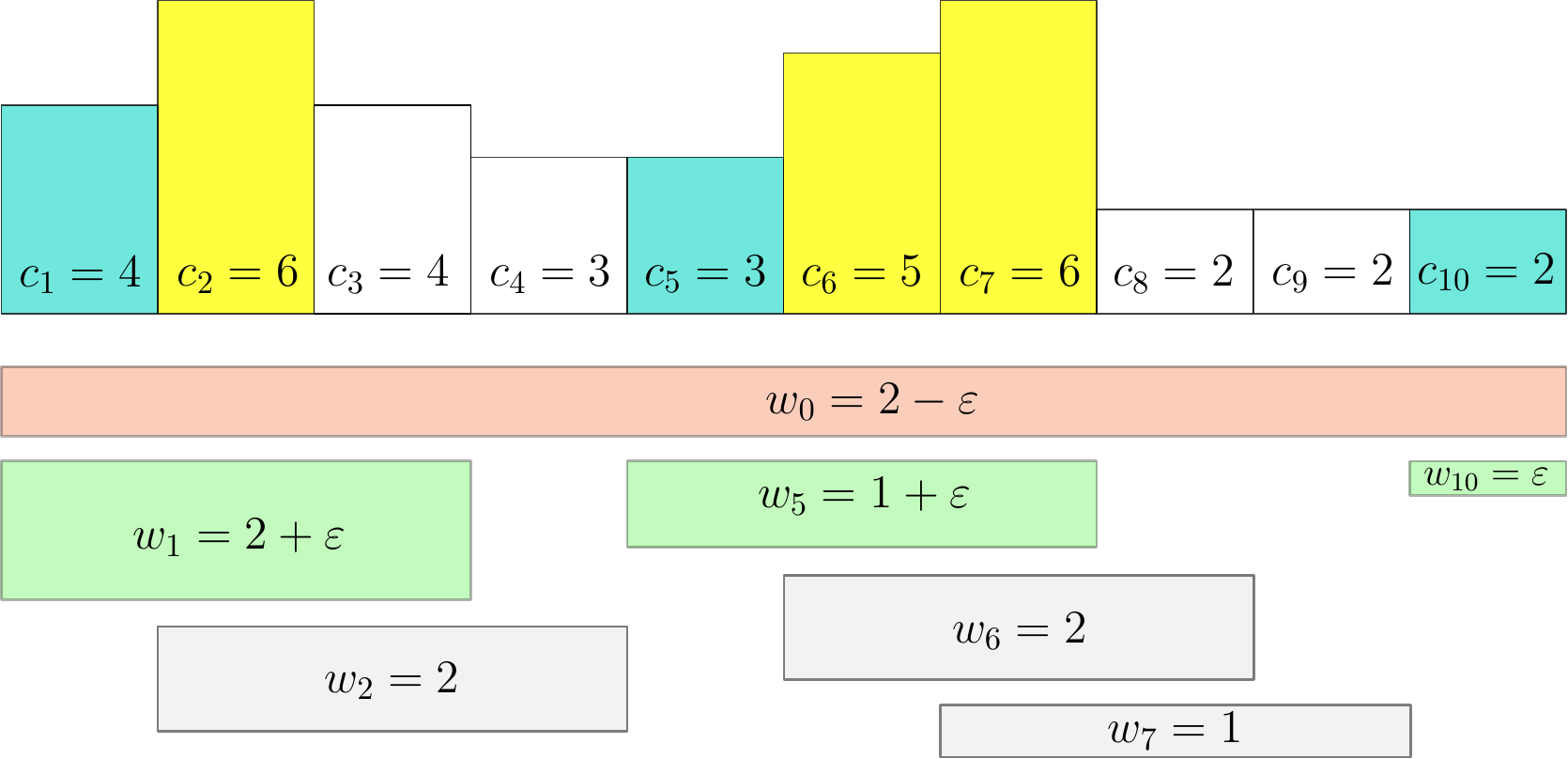}}
        \end{subfigure}
        \caption{A sketch of an instance of $\LPdeltadual$: $c = (4, 6, 4, 3, 3, 5, 6, 2, 2, 2)$ and $\Delta = 3$. The left and right figure depicts $w$ obtained by $\dualgreedy(c, 2)$ and $\dualgreedy(c, 2 - \vareps)$ for $0 < \vareps < 1$, respectively. Constraints $1$, $5$, and $10$ are active, i.e., if $w_0$ on the left is decreased by $\vareps$ then only $w_1$, $w_5$, and $w_{10}$ increase by $\vareps$, as shown on the right. Every $w_i$, for $i \ge 1$, covers $\Delta$ $c$-poles to its right.
        }
        \label{fig:active-constraints}
    \end{center}
\end{figure*}
Next we introduce the key concept that we need for relating the solution of dual to the sparsity of Lagrangian relaxation of the primal: the notion of \emph{active constraints}.
Let $w$ be a vector obtained by $\dualgreedy(c, \lambda)$ and let $w'$ be a vector obtained by $\dualgreedy(c, \lambda - \vareps)$, for some fixed $\lambda \in \bbZ$ and small $\vareps \in (0, 1)$. Then, the set of coordinates that are for $\vareps$ larger in $w'$ than in $w$ are called active constraints. Figure~\ref{fig:active-constraints} provides an illustration of this concept. Intuitively, the active constraints correspond to those variables of $\LPdeltadual$ that increase when $w_0$ decreases by some small value. Hence, one can interpret active constraints as gradients of $\LPdeltadual_{\lambda}$ with respect to the variable $\lambda$. 
This concept appears to be very useful in characterizing the optimal solution of $\LPdeltadual$ in an alternative way. In particular, the following lemma holds.
\begin{lemma}\label{lemma:active-sets-optimality-simplified}
    Let $c \in \bbZ^d$, $\lambda \in \bbZ$, and $w \gets \dualgreedy(c, \lambda)$. Then, if $w$ has exactly $k$ active constraints the vector $w$ is an optimal solution to $\LPdeltadual$.
\end{lemma}
A full proof of a statement stronger than Lemma~\ref{lemma:active-sets-optimality-simplified} along with its proof appears in Lemma~\ref{lemma:active-sets-optimality}, Appendix~\ref{app:active-constraints}, while in this section we provide a proof sketch. Let $w(\vareps) = \dualgreedy(c, \lambda - \vareps)$, for some small $\vareps \in (0, 1)$. By the definition of active constraints and the fact that $w$ has $k$ many, there are exactly $k$ coordinates that are larger by $\vareps$ in $w(\vareps)$ than in $w$. In addition, $w(\vareps)_0 = \lambda - \vareps$ and $w_0 = \lambda$. It is not hard to show that all the other coordinates of $w(\vareps)$ and $w$ are the same, which we can express as $\sum_{i = 1}^d w(\vareps)_i = k \vareps + \sum_{i = 1}^d w_i$. Now, recall that the objective function of dual $\LPdeltadual$ with respect to vector $w$ equals $w_0 k + \sum_{i = 1}^d w_i$. Then we have
\[
    w_0 k + \sum_{i = 1}^d w_i = (w_0 - \vareps) k + k \vareps + \sum_{i = 1}^d w_i = w(\vareps)_0 k + \sum_{i = 1}^d w(\vareps)_i.
\]
Hence, the objective values of dual $\LPdeltadual$ for vectors $w(\vareps)$ and $w$ are equal, for \textbf{all} the values $\vareps \in (0, 1)$. As $\LPdeltadual$ is convex in the value of variable $w_0$ and $\dualgreedy(c, \lambda)$ provides an optimal solution to $\LPdeltadual$ such that $w_0 = \lambda$, then $w$ is an optimal solution to $\LPdeltadual$.

\subsection{Wrapping up -- perturbation and optimal sparsity}
\label{sec:perturbation-and-sparsity}

Now we use Lemma~\ref{lemma:active-sets-optimality-simplified} to prove the following, which essentially justifies line~\ref{line:optimal-hc} of $\mainalgorithm$.
\begin{lemma}\label{lemma:from-dual-to-greedy-on-primal}
    Let $c \in \bbZ^d$, $\lambda \in \bbZ$, and $\tw \gets \dualgreedy(c, \lambda)$. Assume that $\tw$ has exactly $k$ active constraints. Then, any optimal support $\Sstar$ of $\projlagr(\lambda - \vareps, c)$, for $0 < \vareps < 1 / d$, has cardinality exactly $k$.
\end{lemma}
So, if we produce $\lambda$ as in Lemma~\ref{lemma:from-dual-to-greedy-on-primal}, we will solve problem~\eqref{eq:proj-massaged}. These steps are implemented by lines~\ref{line:main-lambda}-\ref{line:main-while} of $\mainalgorithm$. However, as illustrated in the beginning of the section, $\lambda$ as in Lemma~\ref{lemma:from-dual-to-greedy-on-primal} might not exist. Intuitively, this situation happens when the number of active constraints jumps from a value smaller than $k$ to a value larger than $k$ for a very small change of $\lambda$. In such a case, we are unable to obtain $\lambda$ as in Lemma~\ref{lemma:from-dual-to-greedy-on-primal}.
A key component of our analysis is showing that there is an efficient way of randomly altering $c$, and obtaining $\tc$, so that with high probability $\tc$ is such that: $\hlambda$ is obtained as at line~\ref{line:main-lambda}; and, $\tw \gets \dualgreedy(\tc, \hlambda)$ has the same property as in Lemma~\ref{lemma:from-dual-to-greedy-on-primal}. Lines~\ref{line:main-X} and~\ref{line:main-perturb} of $\mainalgorithm$ implement this random perturbation. Intuitively, the perturbation achieved by $X$ variables adds noise to our input instance so that the number of active constraints changes by at most one as $\lambda$ slides over the integer domain.
Following this intuition we obtain a proof of Lemma~\ref{lemma:a-single-iteration}. However, due to the space limitation, we present its proof in Appendix~\ref{app:perturbation-and-sparsity}.

We conclude the section by giving a proof of Lemma~\ref{lemma:from-dual-to-greedy-on-primal}.
\begin{proofof}{Lemma~\ref{lemma:from-dual-to-greedy-on-primal}}
Recall that by our assumption there are exactly $k$ active constraints defined by $\tw$. Then from Lemma~\ref{lemma:active-sets-optimality-simplified} it follows that $\tw$ is a minimizer of $\LPdeltadual$, i.e. $\val{\LPdeltadual_{\lambda}} = \val{\LPdeltadual}$. By the integrality of $\LPdeltadual$ we have that $\val{\LPdeltadual} \in \bbZ$. Let $\lambda' = \lambda - \vareps$, for some $0 < \vareps < 1/d$. Then, it holds $\val{\LPdeltadual_{\lambda'}} = \val{\LPdeltadual}$ as only the $k$ variables correponding to active constraints increased by $\vareps$ while $w_0 = \lambda'$ decreased by $\vareps$. From the strong duality we also have $\val{\LPdeltaprimalnok(\lambda')} = \val{\LPdeltadual_{\lambda'}} =  \val{\LPdeltadual}$.

Let $\tu$ be an integral optimal solution of $\LPdeltaprimalnok(\lambda')$. Following the definition we have
\[
    \val{\LPdeltaprimalnok(\lambda')} = (c^T - \lambda \allones^T) \tu + \lambda k + \vareps (\allones^T \tu - k).
\]
Now we have the following properties: $\tu$ is integral; $0 < \vareps |\allones^T \tu - k| < 1$ whenever $1 \le |\allones^T \tu - k| \le d$; $c \in \bbZ^d$; $v \in \bbZ$; and $\val{\LPdeltaprimalnok(\lambda')} \in \bbZ$. Therefore, we have $\allones^T \tu - k = 0$, and hence $\allones^T \tu = k$. Since we showed the equivalence between $\LPdeltaprimalnok(\lambda')$ and $\projlagr(\lambda', c)$ the lemma follows.
\end{proofof}

\section{Experiments}
\label{sec:experiments}
\pgfplotsset{ignore legend/.style={every axis legend/.code={\renewcommand\addlegendentry[2][]{}}}}

\newlength{\plotwidth}
\setlength{\plotwidth}{3.5cm}
\newlength{\plotheight}
\setlength{\plotheight}{2.5cm}
\newlength{\plotxspacing}
\setlength{\plotxspacing}{0.9cm}

\pgfplotsset{resultplot/.style={%
  scale only axis,
  enlarge x limits=false,
  grid=major,
  grid style={dotted,gray,thin},
  legend cell align=left,
  legend transposed=true,
  legend columns=-1,
  legend style={anchor=north,at={(1.2,-0.3)}},
  every axis plot/.append style={line width=1pt,mark size=2pt},
  cycle list={{blue,mark=square},{red,mark=asterisk},{olive,mark=+},{blue,mark=o}}}}
\pgfplotsset{resultplot1/.style={%
  resultplot,
  height=\plotheight,
  width=\plotwidth}}
  
 \pgfplotsset{resultprob/.style={%
  resultplot1,
  ylabel=Prob. of recovery,
  xlabel=Measurements $n$}}

\pgfplotsset{timeplot/.style={%
  resultplot1,
  error bars/x dir=none,
  error bars/y dir=both,
  error bars/y explicit,
  xlabel=Problem size $n$,
  ylabel=Running time (s)}}

\pgfplotsset{timeratioplot/.style={%
  resultplot1,
  xlabel=Problem size $n$,
  ylabel=Speed-up)}}
  
\pgfplotsset{timeratioplotk/.style={%
  resultplot1,
  xlabel=Sparsity $k$,
  ylabel=Speed-up (DP\, /\, $\mainalgorithm$)}}
  
\pgfplotsset{timeratioplotd/.style={%
  resultplot1,
  xlabel=Size $d$ (in $10^4$),
  ylabel=Speed-up of $\mainalgorithm$}}

\pgfplotsset{timeratioplotdk/.style={%
  resultplot1,
  xlabel=Size $d$ (in $10^2$),
  ylabel=Speed-up of $\mainalgorithm$}}

\pgfplotsset{timeplotcosamp/.style={%
  resultplot1,
  xlabel=Size $d$ (in $10^2$),
  ylabel=Running time (s)}}

\pgfplotsset{timeplotfixeds/.style={%
  resultplot1,
  xlabel=Size $d$,
  ylabel=Running time (s $\cdot 10^2$)}}

 \pgfplotsset{spikeplot/.style={%
  resultplot1}}
  
\begin{figure*}[t!]
    \begin{center}
        \begin{subfigure}[b]{0.31 \textwidth}
            \begin{tikzpicture}
            \begin{axis}[title=(a), timeratioplotdk,name=timeplot1,ymax=150,ymin=1,anchor=north west,xshift=\plotxspacing,ignore legend]
                \addplot + [mark=none] table[x=d, y=time_ratio] {synthetic_DP_vs_Lag_grow_ratio.txt};
                \legend{}
            \end{axis}
            \end{tikzpicture}
        \end{subfigure}
        \hfill
        \begin{subfigure}[b]{0.31 \textwidth}
            \begin{tikzpicture}
            \begin{axis}[title=(b),timeplotcosamp,name=timeplot1,ymax=160.0,ymin=0.0,legend pos=north west, legend columns=3]
            \addplot table[x=d, y=time_mean] {synthetic_CoSaMP_grow_time_lagrange.txt};
            \addplot table[x=d, y=time_mean] {synthetic_CoSaMP_grow_time_DP.txt};
            \addplot table[x=d, y=time_mean] {synthetic_CoSaMP_grow_time_no-model.txt};
            \legend{$\mainalgorithm$, DP, no model}
            \end{axis}
            \end{tikzpicture}
        \end{subfigure}
        \hfill
        \begin{subfigure}[b]{0.31 \textwidth}
            \begin{tikzpicture}
            \begin{axis}[title=(c),timeplotfixeds,name=timeplot1,ymax=160.0,ymin=0.1,legend pos=north west, legend columns=2, xmode=log, ymode=log]
            \addplot + [mark=none] table[x=d, y=time_mean] {synthetic_DP_vs_Lag_fixed_s_Lag_time_lagrange.txt};
            \addplot + [mark=none] table[x=d, y=time_mean] {synthetic_DP_vs_Lag_fixed_s_DP_time_DP.txt};
            \legend{$\mainalgorithm$, DP}
            \end{axis}
            \end{tikzpicture}
        \end{subfigure}

        \begin{subfigure}[c]{0.31 \textwidth}
            \begin{tikzpicture}
            \begin{axis}[xmajorticks=false,title=(d),spikeplot,name=timeplot1,ymax=0.25,ymin=-0.25]
            \addplot + [mark=none] table[x=x, y=y] {real_data.txt};
            \legend{}
            \end{axis}
            \end{tikzpicture}
        \end{subfigure}
        \hfill 
        \begin{subfigure}[d]{0.31 \textwidth}
            \begin{tikzpicture}
            \begin{axis}[xmajorticks=false,title=(e),spikeplot,name=timeplot1,ymax=0.25,ymin=-0.25]
            \addplot + [mark=none] table[x=x, y=y] {real_data_CoSaMP.txt};
            \legend{}
            \end{axis}
            \end{tikzpicture}
        \end{subfigure}
        \hfill 
        \begin{subfigure}[e]{0.31 \textwidth}
            \begin{tikzpicture}
            \begin{axis}[xmajorticks=false,title=(f),spikeplot,name=timeplot1,ymax=0.25,ymin=-0.25]
            \addplot + [mark=none] table[x=x, y=y] {real_data_CoSaMP_Lagrangian.txt};
            \legend{}
            \end{axis}
            \end{tikzpicture}
        \end{subfigure}
        \caption{Separated sparsity experiments. In the top row we plot running times of our algorithm $\mainalgorithm$ relative to the previous work. Plots (a) and (c) are obtained by projecting signals on the separated sparsity model. Plot (b) compares the running times of CoSaMP with three different projection operators. The variant "no-model" makes no structural assumptions and uses hard thresholding as projection operator. In the bottom row, in (e) we plot the recovery of the signal (d) obtained by the algorithm "no model". Plot (f) shows the recovery of $\mainalgorithm$.
        The both procedures use the same number of measurements.
        }
        \label{figure:train-spikes}
    \end{center}
\end{figure*}
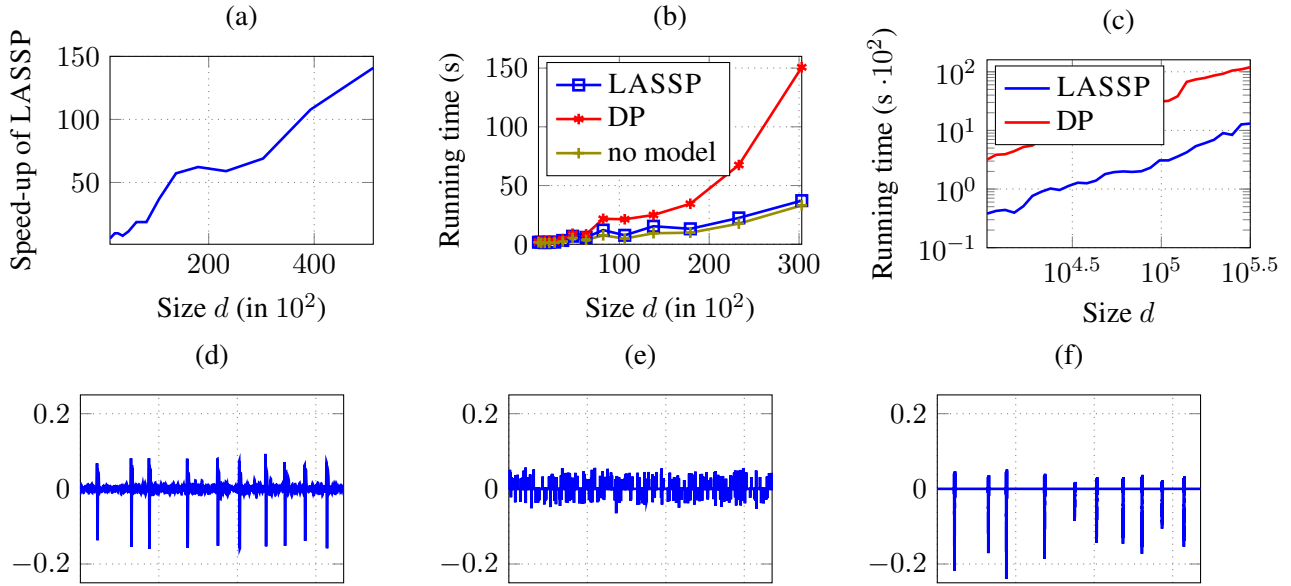

We empirically validate the claims outlined in the previous sections. To that end,
we compare $\mainalgorithm$ with the $O(d k)$-time dynamic program (DP) described in Section~\ref{sec:dp} as a baseline.
Note that this DP already has a better time complexity than the best previously published algorithm from \cite{FMN15}.
Both algorithms are implemented in the Julia programming language (version 0.5.0), which typically achieves performance close to C/C++ for combinatorial algorithms.

\subsection{Synthetic data}
We perform experiments with synthetic data in order to investigate how the algorithms scale as a function of the input size.
We study two different setups: (i) the running time of the projection algorithms on their own, and (ii) the overall running time of a sparse recovery algorithm using the projection algorithms as a subroutine.
For the latter, we use the structure-aware variant of the popular CoSaMP algorithm \cite{NT09, BCDH10}.

\paragraph{Figures~\ref{figure:train-spikes}(a)-(b).}
Given a problem size $d$, we set the sparsity to $k = d / 50$ and generate a random separated sparse vector with parameter $\Delta = (d - 5 (k + 1)) / k - 1$.
The non-zero coefficients are i.i.d.\ $\pm 1$.
For the projection-only benchmark, we add Gaussian noise with $\sigma = 1/10$ to all coordinates in order to make the problem non-trivial.
For each problem size, we run $10$ independent trials and report their mean.

Figure~\ref{figure:train-spikes}(a) shows the speed-up obtained by our nearly-linear time projection relative to the DP baseline.
We observe that $\mainalgorithm$ is up to $150 \times$ faster.
This confirms our expectation that $\mainalgorithm$ scales gracefully with the problem size, while the DP essentially becomes a quadratic-time algorithm.

Figure~\ref{figure:train-spikes}(b) compares the running times of CoSaMP with three different projection operators.
The first variant makes no structural assumptions and uses hard thresholding as projection operator.
The other two variants use a projection for the separated sparsity model, relying on the DP baseline and $\mainalgorithm$, respectively.  
The results show that the version of CoSaMP using $\mainalgorithm$ instead of the DP is significantly faster.
Moreover, CoSaMP with a simple sparse projection has similar running time to CoSaMP with our structured projection.
Finally, we note that CoSaMP with separated sparsity requires $1.5 \times$ fewer measurements to achieve the same recovery quality as CoSaMP with standard sparsity.

\paragraph{Figure~\ref{figure:train-spikes}(c).} We fix the sparsity $k = 100$ and vary the length $d$ of the signal. The signal is obtained in the same way as for plots Figure~\ref{figure:train-spikes}(a)-(b). We observe that our algorithm runs 10x faster than the baseline. Furthermore, the plot shows that the running times of both the baseline DP and $\mainalgorithm$ scale linearly with the signal length $d$.
This behavior is expected for the DP.
On the other hand, our theoretical findings predict that the running time of $\mainalgorithm$ scales as $d \log{d}$.
This suggests that the empirical performance of our algorithm is even (slightly) better than what our proofs state.
\\
In Appendix~\ref{sec:additional-experiments} we report results of additional experiments.

\subsection{Neuronal signals}
We also test our algorithm on neuron spike train data from \cite{HB11}.
See Figure~\ref{figure:train-spikes}(d) for this input data.
First, we run CoSaMP with a ``standard sparsity'' projection.
The recovered signal is depicted in Figure~\ref{figure:train-spikes}(e).
Next, we run the convolutional sparsity CoSaMP of \cite{HB11} and use our fast projection algorithm.
The recovered signal is given in Figure~\ref{figure:train-spikes}(f). For the both experiments we use $n = 250$ measurements.

We also run the convolutional sparsity CoSaMP on neuron spike train data of length $10^5$, comparing the running time of $\mainalgorithm$ and DP as projection operators. CoSaMP with our algorithm runs $2 \times$ faster in this context. We do not compare the running time relative to CoSaMP with a standard sparsity projection as it requires $10\times$ more measurements to achieve accurate recovery.
\section{Conclusions}
We have designed a nearly-linear time algorithm for projecting onto the set of separated sparse vectors.
The core technique in our algorithm is Lagrangian relaxation. One of the key insights here is that even though there are separated sparsity instances for which the Lagrangian relaxation does not provide an optimal solution, it is still possible to obtain an optimal solution if the original input instance is only slightly perturbed.Furthermore, this perturbation does not change the final output, but rather drives the algorithm to choose an optimal solution of interest even after the hard sparsity constraint is relaxed.
Our experiments show that our algorithm is not only of theoretical significance, but also outperforms the state of the art in practice.
Exploring the power of perturbed Lagrangian relaxations for other non-convex constraint sets is an important direction for future work.
We believe that our framework will enable simple and efficient algorithms for other problems as well.

\paragraph{Acknowledgments.} We thank Arturs Backurs for insightful discussions. L. Schmidt thanks Chinmay Hegde for providing a dataset for some of our experiments. A. Mądry was supported in part by an Alfred. P. Sloan Research Fellowship, Google Research Award and the NSF grant CCF-1553428. S. Mitrovi{\' c}  was supported by Swiss NSF (grant number P1ELP2\_161820). Part of this work was carried out while S. Mitrovi{\' c} was visiting MIT.
\newpage

\appendix
\bibliographystyle{IEEEtran}
\bibliography{references}

\onecolumn

\section{Additional experiments}
\label{sec:additional-experiments}
In this section we present experiments additional to those present in Section~\ref{sec:experiments}.

\pgfplotsset{timeplotgrowreal/.style={%
  resultplot1,
  xlabel=Size $d$ (in $10^3$),
  ylabel=Speed-up of $\mainalgorithm$}}
  
\pgfplotsset{timeplotfixedd/.style={%
  resultplot1,
  xlabel=Sparsity $k$,
  ylabel=Running time (s $\cdot 10^2$)}}

\begin{figure*}[t!]
    \begin{center}
        \begin{subfigure}[b]{0.31 \textwidth}
            \begin{tikzpicture}
            \begin{axis}[title=(a), timeplotgrowreal,name=timeplot1,ymax=200,ymin=1,anchor=north west,xshift=\plotxspacing,ignore legend]
                \addplot + [mark=none] table[x=d, y=time_ratio] {real_DP_vs_Lag_grow_ratio_lagrange.txt};
                \legend{}
            \end{axis}
            \end{tikzpicture}
        \end{subfigure}
        \hfill
        \begin{subfigure}[b]{0.31 \textwidth}
            \begin{tikzpicture}
            \begin{axis}[title=(b), timeplotgrowreal,name=timeplot1,ymax=15,ymin=1,anchor=north west,xshift=\plotxspacing,ignore legend]
                \addplot + [mark=none] table[x=d, y=time_ratio] {real_CoSaMP_DP_vs_Lag_grow_ratio_lagrange.txt};
                \legend{}
            \end{axis}
            \end{tikzpicture}
        \end{subfigure}
        \hfill
        \begin{subfigure}[b]{0.31 \textwidth}
            \begin{tikzpicture}
            \begin{axis}[title=(c),timeplotfixedd,name=timeplot1,ymax=500.0,ymin=0.0,legend pos=north west, legend columns=2, xmode=log, ymode=log]
            \addplot + [mark=none] table[x=d, y=time_mean] {synthetic_DP_vs_Lag_fixed_d_Lag_time_lagrange.txt};
            \addplot + [mark=none] table[x=d, y=time_mean] {synthetic_DP_vs_Lag_fixed_d_DP_time_DP.txt};
            \legend{$\mainalgorithm$, DP}
            \end{axis}
            \end{tikzpicture}
        \end{subfigure}
        \caption{The plots compare the running times of our algorithm $\mainalgorithm$ and the DP baseline. Plots (a) and (c) are obtained by projecting signals directly on the separated sparsity model. Plot (b) shows the speed-up of CoSaMP recovery obtained by using $\mainalgorithm$ as a projection operator relative to using the DP for projection.
        }
        \label{figure:appendix-experiments}
    \end{center}
\end{figure*}
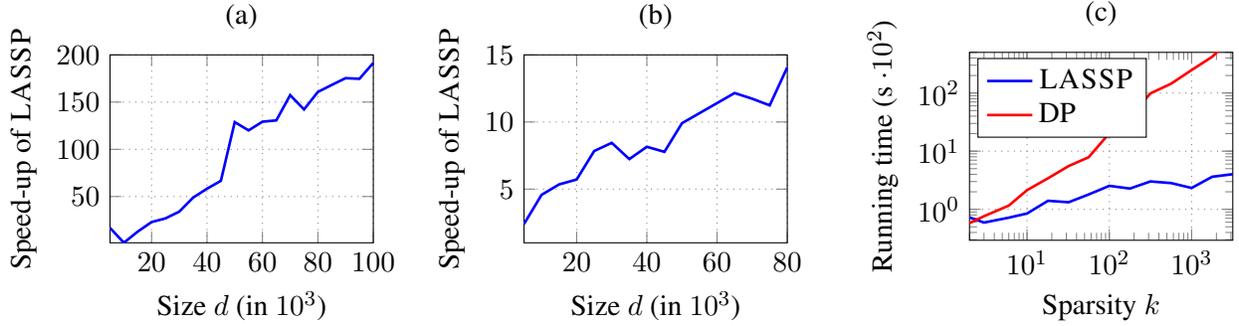

\begin{figure*}[t!]
    \begin{center}
        \begin{subfigure}[b]{0.31 \textwidth}
            \begin{tikzpicture}
            \begin{axis}[title=(a),timeplotcosamp,name=timeplot1,ymax=160.0,ymin=0.0,legend pos=north west, legend columns=3]
            \addplot table[x=d, y=time_mean] {synthetic_App1_CoSaMP_grow_Lag_time_lagrange.txt};
            \addplot table[x=d, y=time_mean] {synthetic_App1_CoSaMP_grow_DP_time_DP.txt};
            \addplot table[x=d, y=time_mean] {synthetic_App1_CoSaMP_grow_no_model_time_no-model.txt};
            \legend{$\mainalgorithm$, DP, no model}
            \end{axis}
            \end{tikzpicture}
        \end{subfigure}
        \hfill
        \begin{subfigure}[b]{0.31 \textwidth}
            \begin{tikzpicture}
            \begin{axis}[title=(b),timeplotcosamp,name=timeplot1,ymax=200.0,ymin=0.0,legend pos=north west, legend columns=3]
            \addplot table[x=d, y=time_mean] {synthetic_App2_CoSaMP_grow_Lag_time_lagrange.txt};
            \addplot table[x=d, y=time_mean] {synthetic_App2_CoSaMP_grow_DP_time_DP.txt};
            \addplot table[x=d, y=time_mean] {synthetic_App2_CoSaMP_grow_no_model_time_no-model.txt};
            \legend{$\mainalgorithm$, DP, no model}
            \end{axis}
            \end{tikzpicture}
        \end{subfigure}
        \hfill
        \begin{subfigure}[b]{0.31 \textwidth}
            \begin{tikzpicture}
            \begin{axis}[title=(c),timeplotcosamp,name=timeplot1,ymax=200.0,ymin=0.0,legend pos=north west, legend columns=3]
            \addplot table[x=d, y=time_mean] {synthetic_App3_CoSaMP_grow_Lag_time_lagrange.txt};
            \addplot table[x=d, y=time_mean] {synthetic_App3_CoSaMP_grow_DP_time_DP.txt};
            \addplot table[x=d, y=time_mean] {synthetic_App3_CoSaMP_grow_no_model_time_no-model.txt};
            \legend{$\mainalgorithm$, DP, no model}
            \end{axis}
            \end{tikzpicture}
        \end{subfigure}
        \caption{The plots compare the running times of CoSaMP with three different projection operators. The variant "no-model" makes no structural assumptions and uses hard thresholding as projection operator. The other two variants use a projection for the separated sparsity model, relying on the DP baseline and $\mainalgorithm$.}
        \label{figure:appendix-experiments-CoSaMP}
    \end{center}
\end{figure*}
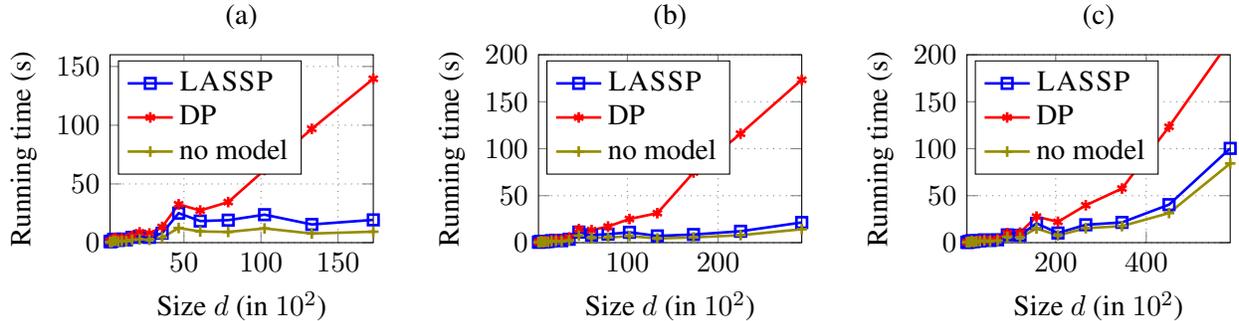
In plots Figure~\ref{figure:appendix-experiments}(a)-(b) we run experiments on real signals. The signals are obtained from signal in Figure~\ref{figure:train-spikes}(d) by concatenating its copies. Figure~\ref{figure:appendix-experiments}(a) is obtained by projecting the signal directly to the separated sparsity model. Then we plot the ratio of the running times of the baseline DP and our algorithm $\mainalgorithm$. Figure~\ref{figure:appendix-experiments}(b) represents the same type of speed-up ratio but this time for the CoSaMP recovery of signal measurements.

For the plots Figure~\ref{figure:appendix-experiments}(c) and Figure~\ref{figure:appendix-experiments-CoSaMP}(a)-(c) we make the setup similar to the one for Figures~\ref{figure:train-spikes}(a)-(b). Namely, given a problem size $d$, parameters $\alpha$ and $\beta$ we set the sparsity to $k = d / \alpha$ and generate a random separated sparse vector with parameter $\Delta = (d - \beta (k + 1)) / k - 1$.
The non-zero coefficients are i.i.d.\ $\pm 1$. We add Gaussian noise with $\sigma$ to all coordinates in order to make the problem non-trivial. For each problem size, we run $10$ independent trials and report their mean.

In Figure~\ref{figure:appendix-experiments}(c) we compare the running times of projections of the baseline DP and our algorithm $\mainalgorithm$ on synthetic data for the following parameters: $d = 5 \cdot 10^4$; varying the sparsity $k$; $\beta = 5$; and $\sigma = 0.5$. This plot confirms our theoretical findings that the running time of $\mainalgorithm$ scales more gracefully with the growth of the sparsity than the running time of DP does.

Similar to Figure~\ref{figure:train-spikes}(b), Figure~\ref{figure:appendix-experiments-CoSaMP}(a)-(c) compare the running times of CoSaMP with three different projection operators. The first variant makes no structural assumptions and uses hard thresholding as projection operator. The other two variants use a projection for the separated sparsity model, relying on the DP baseline and $\mainalgorithm$, respectively. We use the following parameters:
\begin{itemize}
    \item Figure~\ref{figure:appendix-experiments-CoSaMP}(a): $\alpha = 10$; $\beta = 2$; and $\sigma = 0.01$
    \item Figure~\ref{figure:appendix-experiments-CoSaMP}(b): $\alpha = 20$; $\beta = 8$; and $\sigma = 0.95$
    \item Figure~\ref{figure:appendix-experiments-CoSaMP}(c): $\alpha = 100$; $\beta = 30$; and $\sigma = 0.2$.
\end{itemize}
The results show that CoSaMP with a simple sparse projection has similar running time to CoSaMP with our structured projection. Moreover, CoSaMP using $\mainalgorithm$ as the structured projection is usually significantly faster than CoSaMP using the baseline DP. We note that even in the favorable case for the DP structured projection, like in Figure~\ref{figure:appendix-experiments-CoSaMP}(c) where the sparsity $k$ is $100\times$ smaller than the dimension $d$, our algorithm again shows better performance.

\section{An approximate projection counterexample}
\label{sec:noapprox}
Multiple recent algorithms for structured sparse projections build on the \emph{approximation-tolerant} framework of \cite{HIS15IT}.
In this framework, it suffices to design approximate projections instead of solving the model projection problem exactly.
For some sparsity structures, this approach has led to significantly faster algorithms \cite{HIS14icalp,HIS15icml}.
Hence it is interesting to see whether the approximation route is also helpful for separated sparsity.
In fact, it is easy to design the following $2$-approximation algorithm that runs in nearly-linear time.

Partition the vector $c$ into blocks of length $\Delta$.
Then, number those blocks 1 through $\lceil d/\Delta \rceil$ in the order they appear in $c$.
In each block, choose an index corresponding to the largest coordinate of $c$ within that block.
Split those indices into two groups $G$ and $H$ based on the parity of the corresponding block. Formally, define $G$ and $H$ as follows
\[
    G := \left\{\arg \max_{j = (i - 1) \Delta + 1 \ldots \min\{i \Delta, d\}} c_j\ |\ 1 \le i \le \lceil d/\Delta \rceil \text{ and $i$ is odd}\right\},
\]
and similarly
\[
    H := \left\{\arg \max_{j = (i - 1) \Delta + 1 \ldots \min\{i \Delta, d\}} c_j\ |\ 1 \le i \le \lceil d/\Delta \rceil \text{ and $i$ is even}\right\}.
\]
Let $g_1, \ldots, g_{|G|}$ be an ordering of the elements of $G$ so that $c_{g_i} \ge c_{g_j}$ whenever $i \le j$.
Similarly, let $h_1, \ldots, h_{|H|}$ be an ordering of the elements of $H$ such that $c_{g_i} \ge c_{g_j}$ whenever $i \le j$.

Next, define $G' := \{g_i\ |\ 1 \le i \le \min\{k, |G|\}\}$ and $H' := \{h_i\ |\ 1 \le i \le \min\{k, |H|\}\}$.
Now, it is not hard to see that the set attaining larger value among $G'$ and $H'$ has a solution value that is at least $OPT / 2$, where $OPT$ is the maximum sum attainable with a $\Delta$-separated and $k$-sparse vector (in the language of \cite{HIS15IT}, this is an approximate head projection).

While the above algorithm runs in nearly-linear time and achieves a constant-factor approximation, there is a catch.
In particular, the returned support pattern might contain fewer than $k$ indices (as a simple example, consider the case $k=2$, $\Delta = 2$, and $c = (1, 100, 1)$).
This raises the question of the sample complexity for this relaxed sparsity model.
Let $k'$ denote the number of output indices of the $2$-approximation algorithm.
As we have described in the preliminaries, the sample complexity of the separated sparsity model is $O(k \log{(n/k - \Delta)})$.
For concreteness, we now consider the case $k' = k/2$ and $d/k - \Delta \in O(1)$ (the latter is the important regime where the separated sparsity model achieves a sample complexity of $m = O(k)$).
Then we have $d/k' - \Delta \in O(1) + \Delta$. Therefore, if we would like to apply the $2$-approximation algorithm in the case $k' = k/2$, the sample complexity would be
\[
    O(k \log{(2 d/k - \Delta)}) \; = \; O(k \log{\Delta}) \; = \; O(k \log d/k)
\]
where we used the assumption that $d/k - \Delta \in O(1)$.
It is important to note that this sample complexity is worse than $O(k)$ and falls back to the $O(k \log d/k)$ sample complexity of ``standard'' $k$-sparse recovery.
So by using an approximate projection we have lost the sample complexity advantage of the separated sparsity model.
Since we already know how to recover $k$-sparse vectors in nearly-linear time without resorting to structured sparsity, the approximation algorithm does not provide a novel trade-off.

\section{Dynamic programming}
\label{sec:dp}
Before we introduce our dynamic programming (DP) algorithm for the separated sparsity problem, we briefly review a variant of the DP given in \cite{FMN15}.
We remark that the variant below has a time complexity of $O(d k)$, which is already an improvement over the $O(d^2)$ of \cite{FMN15} when $k = o(d)$. However, the authors are not aware of a work the describes the $O(d k)$ DP approach, and hence refer to it as folklore.
We improve the analysis of this DP further and give a faster variant that runs in time $O(k (d - (k - 1) \Delta))$.
In the regime where where the slack $d / k - \Delta$ is constant, the running time simplifies to $O(d + k^2)$.
So for the very sparse case $k \in \Otilde(\sqrt{d})$, our improved DP already achieves a nearly-linear running time.

\subsection{The basic dynamic program}
The folklore dynamic program fills the following table $\dpf[i][j]$.
Value $\dpf[i][j]$ can be interpreted as follows: Maximal value of choosing $j$ coordinates of $c$ with positions in $1 \ldots i$, such that any two chosen coordinates are at distance $\Delta$ at least.
\begin{alignat*}{4}
	\dpf[i][0]   & = 0 &  & \\
	\dpf[i][j]   & = -\infty, &  & \text{ if } i < 1 \\
	\dpf[i][j]   & = \max\left\{\dpf[i - 1][j], c_i + \dpf[i - \Delta][j - 1]\right\},  & & \text{ otherwise }
\end{alignat*}
The following claim follows easily:
\begin{theorem}[Folklore]\label{theorem:dp-folklore}
    $\dpf[n][k]$ is an optimal value to the $\Delta$-separated problem and can be computed in time $O(d k)$. The elements that constitute the optimal value can be output from $\dpf$ in time $O(d + k)$.
\end{theorem}

\subsection{An improved dynamic program}
The dynamic program outlined above already achieves an improvement over the algorithms in \cite{hegdeLP,FMN15}.
However, the dynamic program is still wasteful with the state space it considers.
Consider case in which there is only one possible configuration, i.e. $d = (k - 1) \Delta + 1$ for some $k, \Delta \ge 1$.
Regardless of the cost vector $c$, an input with such parameters has only one valid solution.
Nevertheless, the natural implementation of the dynamic program above runs in $\Theta(d k)$ time.
Note that the definition of $\dpf$ does not take into account the \emph{mandatory distance} required by the separated sparsity model.  
By mandatory we refer to $\Delta$ distance between any two chosen coordinates of $c$.
This observation gives raise to new dynamic programming definition, that we denote by $\dpfinal$, which directly implements mandatory distances. Let
\[
    s := d - ((k - 1) \Delta + 1),
\]
which can be seen as the \emph{slack} distance that is not mandatory. For instance, if we choose two coordinates of $c$ with indices $i$ and $i + \Delta + 5$, and no other coordinate between them, then we say that $\Delta$ distance between them is mandatory, and the remaining distance of 5 is slack. Now, $\dpfinal$ is defined as
\begin{alignat*}{4}
	\dpfinal[i][0]   & = 0 &  & \\
	\dpfinal[i][j]   & = -\infty, &  & \text{ if } i < 0 \\
	\dpfinal[i][j]   & = \max\left\{\dpfinal[i - 1][j], c_{pos(i, j)} + \dpfinal[i][j - 1]\right\},  & & \text{ otherwise }
\end{alignat*}
where $pos(i, j) = d - (s - i) - (k - j) \Delta$. Value $pos(i, j)$ gives us the smallest index of $c$ that we can chose if to the left of that index there are $j$ other chosen indices, and slack distance of $i$ is used. With the definition of $\dpfinal$ in hands one can easily show the following theorem.
\begin{theorem}\label{theorem:dp-final}
    Value $\dpfinal[s][k]$ is an optimal value to the $\Delta$-separated problem and can be computed in time $O(k s)$. The elements that constitute the optimal value can be output from $\dpfinal$ in time $O(d + k)$.
\end{theorem}
Comparing Theorem~\ref{theorem:dp-folklore} and Theorem~\ref{theorem:dp-final}, we see that for certain input parameters, e.g. when $s \in O(d)$, we did not make much progress by devising $\dpfinal$. However, as we have already mentioned, the setting we care the most is when $d/k - \Delta$ is a constant, as with such a choice of parameters the sample complexity is $O(k)$. But in that case we do achieve a significant improvement over Theorem~\ref{theorem:dp-folklore}.
\begin{corollary}
    If $d / k - \Delta \in O(1)$, then there is an algorithm that solves the $\Delta$-separated problem in $O(k^2 + d)$ time.
\end{corollary}
\section{Two dimensional $\Delta$-separated problem is NP-Hard}
\label{app:twod}

In this section we consider a natural extension of the $\Delta$-separated problem in which vector $c$ is two dimensional, i.e. $c \in \bbR^{d \times m}$. Formally, given $c$, sparsity parameter $k$, and integer $\Delta$, the goal is to output $k$ pairs of integers $(i^t, j^t)$ such that:
\begin{itemize}
    \item $1 \le i^t \le d$, and $1 \le j^t \le m$, for all $t = 1 \ldots k$;
    \item for every $t \neq s$ we have $\min\left\{\left|i^t - i^s\right|, \left|j^t - j^s\right|\right\} \ge \Delta$; and
    \item $\sum_{t = 1}^{k} c_{i^t, j^t}$ is maximized.
\end{itemize}
We refer by $\twoddelta(c, k, \Delta)$ to that problem, and show it is NP-hard by reducing it to problem $\boxpack$ studied in~\cite{2d-delta-separation}. Let us start by recalling the definition of $\boxpack$.

Let $k$ be a number of identical $3 \times 3$ squares.\footnote{We choose sides to be of length 3 as already that setting is sufficient to prove NP-hardness of $\boxpack$, see the proof of Theorem~2 in~\cite{2d-delta-separation}.} Our goal is to pack the $k$ squares into a region of plane defined by set $R$. Set $R$ consists of pairs of integers, where every pair represents the point at which a square can be placed, e.g. the upper-left corner of a square. Every square can only be placed so that its sides are parallel to the axis. By $\boxpack(k, R)$ we refer to the problem of answering whether for given $k$ and $R$ one can place all the $k$ squares in the described way, such that no two squares overlap. In~\cite{2d-delta-separation} is proved the following result.
\begin{theorem}[Theorem~2 in~\cite{2d-delta-separation}]\label{theorem:boxpack-NP-hard}
    $\boxpack$ is $NP$-complete.
\end{theorem}
Now we utilize Theorem~\ref{theorem:boxpack-NP-hard} to show that $\twoddelta(c, k, \Delta)$ is NP-hard as well.
\begin{theorem}
    Problem $\twoddelta(c, k, \Delta)$ is NP-hard.
\end{theorem}
\begin{proof}
We provide a polynomial time reduction of $\boxpack(k, R)$ to a sequence of $\twoddelta$ instances.

Let $R^E$, denoting "extended $R$", be the set of all integers points that can be occupied by a square, not necessarily its corners. Observe that $\left|R^E\right| \le 9 |R|$. We say that two points in $R^E$ are adjacent if they share the same $x$- or the same $y$-coordinate. Consider a connected component $F$ of $R^E$.

Assume that $\boxpack(i, F)$ can be reduced to an $\twoddelta$ instance. Then, let us show that $\boxpack(k, R)$ can be reduced to polynomially many $\twoddelta$ instances.

Let $\cF$ be the set of all connected components of $R^E$. Then, for every $C \in \cF$ define $val_F$ as follows
\[
    val_F := \arg\max_{i \in \{0, \ldots, \min\{|F|, k\}\}} \boxpack(i, F) \text{ equals true}.
\]
Now, clearly, if $\sum_{F \in \cF} val_F \ge k$, then $\boxpack(k, R)$ equals true. Observe that in order to compute all $val_F$, we need to compute only polynomially many instances of $\boxpack$, where all the input parameters are bounded by $k$ and $R$. So, it remains to reduce $\boxpack(i, F)$ to an instance of $\twoddelta$.

To that end, consider $F \in \cF$. Let $h$ and $w$ be the smallest values such that $F$ is enclosed by an axis-parallel sides rectangle $T$ of height $h$ and width $w$. With loss of generality, assume that $T$ and $F$ are translated so that the bottom-left corner of $T$ at $(1, 1)$. Define vector $c \in \{0, 1\}^{h \times w}$ to be an indicator vector of points of $C$ at which can be placed square, i.e.
\[
    c_{i, j}^F =
    \begin{cases}
        1 & \text{if } 3 \times 3 \text{ square with upper-left corner at } (i, j) \text{ is in } F \\
        0 & \text{otherwise}
    \end{cases}
\]
Note that $h w \le |F|^2$. Next, define $opt_F$ as follows
\[
    opt_F := \max_{i \in \{1, \ldots, \min\{|F|, k\}\}} \twoddelta(c^F, i, 3).
\]
Then, we claim $opt_F$ equals $val_F$. Now, it is easy to see that the claim is true. Value $val_F$ represents the maximal number of squares that can be packed within $F$. Each such square corresponds to an entry of $c^F$ which has value $1$. So, we have $\twoddelta(c^F, val_F, 3) = val_F$, and hence $opt_F \ge val_F$. On the other hand, by the construction of $c^F$, $opt_F$ represents a number of squares (not necessarily maximum, though) that can be packed in $F$. So, we have $val_F \ge opt_F$, and therefore $opt_F = val_F$ as claimed.

This concludes our proof.
\end{proof}
\section{Uniform size neuronal spike trains}
\label{app:blocks}
We propose a generalization of the $\Delta$-separated model. In this model, we assume that neuronal spike trains correspond to blocks of uniform size. We use $b$ to denote the size of blocks. Formally, we define model $\Mmodel_{k,\Delta, b}$ to be the set of all vector in $\bbR^d$ such that non-zero entries are grouped in blocks of size $b$, there are exactly $k$ blocks, and every two blocks are separated by at least $\Delta$.

In order to apply the existing framework for recovering measured vectors that belong to $\Mmodel_{k,\Delta, b}$, we develop a method that solves the following projection problem. Given vector $c \in \bbR^d$, the goal is to output $k$ indices $p_1, \ldots, p_k$ such that
\begin{itemize}
    \item $1 \le p_i \le d - b + 1$, for every $i = 1 \ldots k$;
    \item $\left| p_i - p_j \right| \ge \Delta + b - 1$, for all $i \neq j$; and,
    \item the sum
        \[
            \sum_{i = 1}^k \sum_{j = 0}^{b - 1} c_{p_i + j}
        \]
        is maximized.
\end{itemize}
We refer to that problem by $(\Delta, b)$-separated. Now, it is easy to show the following claim.
\begin{theorem}\label{theorem:delta-b}
    Given vector $c \in \bbN^d$, sparsity $k$, and block size $b$, define vector $c^b \in \bbN^{d - b + 1}$ as
    \[
        c^b_i := \sum_{j = i}^{i + b - 1} c_j.
    \]
    Let $S$ be a set of indices corresponding to an optimal solution for $(\Delta + b - 1)$-separated problem for cost vector $c^b$ and sparsity $k$. Then, $S$ is an optimal solution to $(\Delta, b)$-separated problem for cost vector $c$ and sparsity $k$.
\end{theorem}
\begin{proof}
    Clearly, $S$ is feasible choice of indices for $(\Delta, b)$-separated problem. Towards a contradiction, assume that there exists another set of indices $S'$ that achieves larger value than $S$.
    
    But then, $S'$ is a feasible choice of indices for $(\Delta + b - 1)$-separated problem. Next, notice that $S$, and also $S'$, achieve the same value for both $(\Delta, b)$- and $(\Delta + b - 1)$-separated problem. However, as $S'$ achieves higher value than $S$, this contradicts our assumption that $S$ is an optimal choice of indices for $(\Delta + b - 1)$-separated problem.
\end{proof}
Now we can solve $(\Delta, b)$-separated problem in nearly-linear time.
\begin{corollary}
    A solution to $(\Delta, b)$-separated problem can be computed in nearly-linear time.
\end{corollary}
\begin{proof}
    We use the reduction from Theorem~\ref{theorem:delta-b} to reduce $(\Delta, b)$-separated problem to $(\Delta + b - 1)$-separated one. The reduction can be applied in linear time. Then, the claim follows by Theorem~\ref{theorem:recover-primal}.
\end{proof}

\subsection{Sample complexity}
Next we analyze the sample complexity of $\Mmodel_{k,\Delta, b}$ model. To that end, we count the number of support in that model. Now, all the supports are captured by $(\alpha_0, \alpha_1, \ldots, \alpha_k)$, where $\alpha_i \ge 0$ and $\sum_{i = 0}^k \alpha_i = d - k b - (k - 1) \Delta$. The idea is that $k$ blocks split the vector in $k + 1$ regions. All but the first and last region correspond to coordinates between two neighboring blocks. Value $\alpha_0$ and $\alpha_k$ correspond to coordinates before the first block and after the last block, respectively. Value $\alpha_i$, for $1 \le i < k$ correspond to the slack distance between block $i$ and $i + 1$.

Let $\Msupports_{k,\Delta, b}$ be the family of support patterns of $\Mmodel_{k,\Delta, b}$. Then, we have:
\[
    \left|\Msupports_{k,\Delta, b}\right| = \binom{d - k b - (k - 1) \Delta + k}{k}.
\]
Now, from the result that there are RIP matrices with $O\left(k + \log{\left| \Msupports_{k,\Delta, b} \right|}\right)$ rows, \cite{BCDH10}, the following claim follows.
\begin{theorem}
    The sample complexity of $\Mmodel_{k,\Delta, b}$ is $O(k \log{(d/k - (b + \Delta) + (\Delta - 1)/k)})$. In particular, if $(\Delta - 1)/k \in O(1)$, then the sample complexity is $O(k \log{(d/k - (b + \Delta)})$.
\end{theorem}
\section{Omitted proofs from Section~\ref{sec:dual} and study of $\LPdeltadual$}
\label{app:dual}
Our running-time results are given with respect to $\gamma$, where $\gamma$ is the maximal number of bits needed to store any $c_i$. As $\LPdeltaprimal$ and $\LPdeltadual$ are invariant under shifting and multiplication of $c$, for the sake of clarity of our exposition and without loss of generality we assume $c$ is an integral vector with non-negative entries.

\begin{restatable}{lemma}{primaltum}
\label{lemma:primal-is-TUM}
    The constraint matrix of problem $\LPdeltaprimal$ is \emph{totally unimodular} (TUM).
\end{restatable}
\begin{proof}
    We first rewrite the constraints of $\LPdeltaprimal$ to be in the form $B u \le b, u \ge 0$, where $B u \le b$ is defined as
    \begin{alignat*}{4}
    	& \qquad & \sum_{i = 1}^d{u_i} & \le k \\
    	& \qquad & \sum_{i = 1}^d{-u_i} & \le -k \\
        & \qquad & \sum_{j = i}^{\min\{i + \Delta - 1, d\}}{u_j} & \le 1 & \qquad & \forall i = 1 \ldots d
    \end{alignat*}
    Let $D$ be a square submatrix of $B$. If we show that $\det{D} \in \{-1, 0, 1\}$, then by the definition of TUM the lemma follows.
    
    Now, $D$ is either a binary \emph{interval matrix}, or it has one row where all the non-zero entries are consecutive and equal $-1$ and the other rows constitute a binary interval matrix. Let $D'$ be a matrix defined as $D_{i, j}' := \left|D_{i, j}\right|$. Then, we have $\left|\det{D}\right| = \left|\det{D'}\right|$. Also, we have that $D'$ is a binary interval matrix. However, it is well known that binary interval matrices are TUM, see~\cite{nemhauser1988integer}. Hence, $\det{D} \in \{-1, 0, 1\}$
\end{proof}

\integraloptimal*
\begin{proof}
    By Lemma~\ref{lemma:primal-is-TUM}, problem $\LPdeltaprimal$ is totally unimodular.
    By~\cite{schrijver2002combinatorial} we have that $\LPdeltadual$ is totally unimodular as well.
    This implies that $\LPdeltadual$ has an integral optimal solution.
\end{proof}

The crucial property of the LP $\LPdeltadual$ is the following: for a fixed value of the variable $w_0$, we can solve $\LPdeltadual$ in \emph{linear} time. Specifically, let us define $\LPdeltadual_v$ to be the LP $\LPdeltadual$ in which the variable $w_0$ is set to $v$. (Note that by Corollary~\ref{corollary:y0-integral-optimal}, we can restrict our attention to integer values of $v$.) 
Observe that once we fixed the value of $w_0$, all remaining constraints in $\LPdeltadual_v$ are ``local'' since they only affect a known interval of length $\Delta$. They are also ordered in a natural way. 
As a result, we can solve $\LPdeltadual_v$ by making a single pass over these variables.
Starting with $w_1$ and all variables set to $0$, we consider each constraint from left to right and increase the variables to satisfy these constraints in a lazy manner. That is, if reach in our pass a constraint with index $i$ that is still not satisfied, we increase the value of $w_i$ until that constraint becomes satisfied and then move to the next constraint and index. Algorithm~\ref{alg:dual}, called $\dualgreedy$, present in Section~\ref{sec:duality-and-back} formalizes this approach.
\begin{restatable}{lemma}{greedyoptimal}
\label{lemma:greedy-is-optimal}
    For any $\mc$, and $\mk$, the algorithm $\dualgreedy(\mc, v)$ computes an optimal solution $\LPdeltadual_v(\mc, \mk)$ in linear time.
\end{restatable}
\begin{proof}
    Let $\mc \in \bbN^\mn$.
    We show that the algorithm $\dualgreedy(\mc, v)$ outputs vector $\tw$ in $O(\mn)$ time such that $\LPdeltadual_v(\mc, \mk)$ is minimized and $\tw_0 = v$.
    To show the running time, observe that every iteration in the for loop takes $O(1)$ time, so the total algorithm runs in $O(\mn)$ time.

    Next, it is easy to see that $\tw$ is a feasible solution to $\LPdeltadual_v(\mc, \mk)$. We show that it is a minimal as well.
    \\
    Towards a contradiction, assume there is a vector $w'$ such that $w'_0 = \tw_0 = v$, $w'$ is a feasible solution to $\LPdeltadual(\mc, \mk)$, and $\|w'\|_1 < \|\tw\|_1$. Then, there exists an index $j$ such that $w'_j \neq \tw_j$ and $w'_i = \tw_i$ for all $i < j$. In case there are multiple vectors $w'$, let $w'$ be one that maximizes the location of mismatch $j$. Consider the two possible cases: $w'_j < \tw_j$, and $w'_j > \tw_j$.
    
    \paragraph{Case $w'_j < \tw_j$.} Observe that $\tw_j$ is chosen as a function of $\tw_0, \ldots, \tw_{j - 1}$ as a minimal value so that $\tw$ is feasible. Therefore, if $w'_j < \tw_j$ and $w'_i = \tw_i$ for all $i < j$, then $w'$ could not be feasible.
    
    \paragraph{Case $w'_j > \tw_j$.} First, note that $j < \mn$, as otherwise $\|w'\|_1 > \|\tw\|_1$. Now, we construct $w''$ as follows. We set $w''_i = w'_i$ for all $i$ different than $j$ and $j + 1$. Set $w''_j = \tw_j$ and $w''_{j + 1} = w'_{j + 1} + \tw_j - w'_j$. Clearly, $w''$ is also a feasible solution to $\LPdeltadual(\mc, \mk)$. Furthermore, $\|w''\|_1 = \|w'\|_1$ and $\tw$ and $w''$ agrees on first $j$ coordinates, contradicting our choice of $w'$.
    
    This concludes the proof.
\end{proof}

Now, to obtain an algorithm that is also able to solve the original dual LP $\LPdeltadual$ (instead of only $\LPdeltadual_v$) it suffices to provide a procedure for choosing the optimal value of $v$. A priori, there can be many possible choices of $v$ and thus an exhaustive search would be prohibitive. Fortunately, $\LPdeltadual_v$ is actually convex in $v$, and so we can use a ternary search over $v$ to find such an optimal value.
\begin{lemma}\label{lemma:LPdeltadual-is-convex}
	$\LPdeltadual_v(\mc, \mk)$ is convex in $v$.
\end{lemma}
\begin{proof}
	Let $\tw$ and $\hw$ be $w$-solution to $\LPdeltadual_{\hv}(\mc, \mk)$ and $\LPdeltadual_{\tv}(\mc, \mk)$, respectively. Note that $\hw_0 = \hv$ and $\tw_0 = \tv$. Now if we show that
	\[
		\LPdeltadual_{\frac{\hv + \tv}{2}}\left(\mc, \mk \right) \le \frac{\LPdeltadual_{\hv}(\mc, \mk) + \LPdeltadual_{\tv}(\mc, \mk)}{2},
	\]
	the convexity will follow. We start by showing that $\tfrac{\hw + \tw}{2}$ is a feasible $w$-vector for the dual, assuming that both $\hw$ and $\tw$ are feasible. We consider the feasibility of each of the constraints.
	\begin{enumerate}[(1)]
		\item From
			\[
				\hw_0 + \sum_{j\ :\ j \le i \le j + \Delta - 1 \text{ and } j \ge 1}{\hw_j} \ge \mc_i
			\]
			and
			\[
				\tw_0 + \sum_{j\ :\ j \le i \le j + \Delta - 1 \text{ and } j \ge 1}{\tw_j} \ge \mc_i
			\]
			we have
			\[
				(\hw_0 + \tw_0) + \sum_{j\ :\ j \le i \le j + \Delta - 1 \text{ and } j \ge 1}{(\hw_j + \tw_j)} \ge 2 \mc_i,
			\]
			which implies
			\[
				\frac{\hw_0 + \tw_0}{2} + \sum_{j\ :\ j \le i \le j + \Delta - 1 \text{ and } j \ge 1}{\frac{\hw_j + \tw_j}{2}} \ge \mc_i.
			\]
		\item Also, from $\hw_j \ge 0$ and $\tw_j \ge 0$ we have $\tfrac{\hw_j + \tw_j}{2} \ge 0$.
		\item Trivially, $\tfrac{\hw_0 + \tw_0}{2} \in \mathbb{R}$.
	\end{enumerate}
	So, indeed $\tfrac{\hw + \tw}{2}$ is feasible. Hence, we have
	\begin{eqnarray*}
		\LPdeltadual_{\frac{\hv + \tv}{2}} & \le & \frac{\hw_0 + \tw_0}{2} \mk + \sum_{i = 1}^{\mn} \frac{\hw_j + \tw_j}{2} \\
				& = & \frac{\left(\hw_0 \mk + \sum_{i = 1}^{\mn} \hw_j\right) + \left(\tw_0 \mk + \sum_{i = 1}^{\mn} \tw_j\right)}{2} \\
				& = & \frac{\LPdeltadual_{\hv}(\mc, \mk) + \LPdeltadual_{\tv}(\mc, \mk)}{2}.
	\end{eqnarray*}
	This completes the proof.
\end{proof}

Putting these pieces together yields the main theorem of this section. As a reminder, we assume that $c$ is an integral vector with non-negative entries.
\begin{theorem}
\label{theorem:ternary-search}
    There exists an algorithm $\optimalvalue(\mc, \mk)$ that, for any $\mc$ and $\mk$, outputs an optimal solution $\wstar$ to $\LPdeltadual(\mc, \mk)$ in time $O\left(\rtoptimalvalue{\mn}{\mcmax}{\mk}\right)$.
\end{theorem}
\begin{proof}
\paragraph{Bounding $\wstar_0$.} Following Lemma~\ref{lemma:LPdeltadual-is-convex}, we can use ternary search to find $\wstar_0$. However, to be able to do that, we have to know the interval $[a, b]$ we are searching over for $\wstar_0$. By the following lemma we give an upper bound on $b$. However, in order to provide a lower bound on $a$, we need to develop some more machinery. So, we defer its proof to later sections, and in Lemma~\ref{lemma:ystar-lower-bound} we show that $a$ can be lower bounded by $ -(k - 1) \cmax$.
\begin{lemma}\label{lemma:ystar-upper-bound}
    Let $\wstar$ be a vector that minimizes $\LPdeltadual$. Then, $\wstar_0 \le \cmax$.
\end{lemma}
\begin{proof}
    Towards a contradiction, let $\wstar$ be an optimal vector such that $\wstar_0 > \cmax$. Now, as $w_i \ge 0$ for all $i \ge 1$, we have that the corresponding objective is at least $k \wstar_0 > k \cmax$. On the other hand, consider vector $\hw$ such that $\hw_0 = \cmax$ and $\hw_i = 0$ for all $i \ge 1$. Clearly, $\hw$ is a feasible solution to $\LPdeltadual$. However, the objective function corresponding to $\hw$ is $k \cmax < k \wstar_0$, which contradicts our assumption that $\wstar$ is a minimizer of $\LPdeltadual$.
\end{proof}

\paragraph{A nearly-linear time algorithm.}
Now we provide an algorithm that computes the optimal value of $\LPdeltadual(\mc, \mk)$ in nearly linear time. It employs ternary search over the interval provided by Lemma~\ref{lemma:ystar-lower-bound} and Lemma~\ref{lemma:ystar-upper-bound} in order to find $w_0$ that optimizes $\LPdeltadual(\mc, \mk)$. At every step of the search, it uses the result from Lemma~\ref{lemma:greedy-is-optimal} to find an optimal solution to $\LPdeltadual_v(\mc, \mk)$, for $v$ chosen at the current search step.
\begin{algorithm}[H]
	\caption{\optimalvalue:}
	\label{alg:optimal-value}
	Input: $\mc \in \bbN^{\mn}$, sparsity $\mk$, lower bound $lb$ (if not specified, the default value is $-(\mk - 1) \mcmax$), upper bound $ub$ (if not specified, the default value is $\mcmax$) \\
	Output: a minimizer $\wbest$ to $\LPdeltadual(\mc, \mk)$ constrained to $\wbest_0 \in [lb, ub]$; if $lb$ and $ub$ are not specified, $\wbest$ is an optimal solution to $\LPdeltadual(\mc, \mk)$
	\vspace{-0.1in}
	\begin{enumerate}[1.]
		\item $s \gets lb$, $e \gets ub$
		\vspace{-0.1in}
		\item While $s \le e$
			\vspace{-0.1in}
			\begin{enumerate}[1.]
				\item $l \gets s + \left \lfloor \tfrac{e - s}{3} \right \rfloor$, $r \gets e - \left \lfloor \tfrac{e - s}{3} \right \rfloor$
				\item $(active^l, w^l) \gets \dualgreedy(c, l)$
				\item $(active^r, w^r) \gets \dualgreedy(c, r)$
				\item If $k w^l_0 + \sum_{i = 1}^{\mn} w^l_i \le k w^r_0 + \sum_{i = 1}^{\mn} w^r_i$ then $e \gets r - 1$, $\wbest \gets w^l$
				\item Else $s \gets l + 1$, $\wbest \gets w^r$
			\end{enumerate}
		\item Return $\wbest$
	\end{enumerate}
\end{algorithm}
Lemma~\ref{lemma:greedy-is-optimal}, Lemma~\ref{lemma:LPdeltadual-is-convex}, Corollary~\ref{corollary:y0-integral-optimal} and bounds on $a$ and $b$, the theorem follows directly.
\end{proof}
\section{A detailed proof of Lemma~\ref{lemma:solve-projlagr}}
\label{app:proj-lagr}

Algorithm $\deltanosparsity$ solves $\LPdeltaprimalnok(\lambda)$ (and so $\projlagr$).
\begin{algorithm}[H]
	\caption{\deltanosparsity:}
	\label{alg:deltanosparsity}
	{\bfseries Input:} $c \in \bbR^{d}$, $\lambda \in \bbR$ \\
	{\bfseries Output:} an integral optimal solution of $\LPdeltaprimalnok(\lambda)$
	\begin{algorithmic}[1]
	    \STATE\label{line:init-no-sparsity} $\tc \gets c - \allones \lambda$
	    \IF{$\tc_1 \ge 0$}\label{line:if-1-no-sparsity}
	        \STATE $best_1 \gets \tc_1$, $pick_1 \gets 1$
	    \ELSE
	        \STATE $best_1 \gets 0$, $pick_1 \gets -1$
	    \ENDIF\label{line:endif-1-no-sparsity}
		\FOR{$i \gets 2 \ldots \Delta$}\label{line:for1-no-sparsity}
			\IF{$\tc_i > best_{i - 1}$}
			    \STATE $best_i \gets \tc_i$, $pick_i \gets i$
            \ELSE
			    \STATE $best_i \gets best_{i - 1}$, $pick_i \gets pick_{i - 1}$
            \ENDIF
        \ENDFOR
		\FOR{$i \gets \Delta + 1 \ldots d$}\label{line:for2-no-sparsity}
            \IF{$\tc_i + best_{i - \Delta} > best_{i - 1}$}
                \STATE $best_i \gets \tc_i + best_{i - \Delta}$, $pick_i \gets i$
			\ELSE
			    \STATE $best_i \gets best_{i - 1}$, $pick_i \gets pick_{i - 1}$
            \ENDIF
        \ENDFOR
		\STATE $r \gets \emptyset$, $i \gets d$
		\WHILE{$i \ge 1$}\label{line:while-no-sparsity}
			\IIF{$pick_i \ge 1$}{$r \gets r \cup \{pick_i\}$}
			\STATE $i \gets pick_i - \Delta$
		\ENDWHILE
        \RETURN $r$
	\end{algorithmic}
\end{algorithm}

\begin{lemma}\label{lemma:full-proof-proflagr}
    Algorithm $\deltanosparsity(c, \lambda)$ solves $\LPdeltaprimalnok(\lambda)$ (and so $\projlagr(\lambda, c)$ in $O(d)$ time.
\end{lemma}
\begin{proof}
    Let us rewrite $\LPdeltaprimalnok(\lambda)$ as follows
    \begin{alignat*}{4}
    	\text{maximize}   &\qquad & (c^T - \lambda \allones^T) u + \lambda k \\
    	\text{subject to} &\qquad & \sum_{j = i}^{\min\{i + \Delta - 1, d\}}{u_j} & \le 1 & \qquad & \forall i = 1 \ldots d \\
    					&\qquad & u_i  & \ge 0 &\qquad & \forall i = 1 \ldots d
    \end{alignat*}
    Observe that for a given $\lambda$ and $k$, the term $\lambda k$ is constant, so our objective becomes $(c^T - \lambda \allones^T) u = \tc^T u$ subject to the given set of constraints. Now we can show that $\deltanosparsity(c, \lambda)$ solves this formulation, which is equivalent to the original one of $\LPdeltaprimalnok(\lambda)$. We proceed by induction. We show that $best_i$ stores the maximum sum of the elements of $\{\tc_1, \ldots, \tc_i\}$ so that any two chosen elements $\tc_i$ and $\tc_j$ are such that $|i - j| \ge \Delta$, and $pick_i$ keeps what is the largest index of element that should be taken to achieve $best_i$ (if there is no such element, then $pick_i$ equals $-1$).
    
    \paragraph{Base of induction, $i = 1$.}
        For $i = 1$, the lines~\ref{line:if-1-no-sparsity}-\ref{line:endif-1-no-sparsity} set $best_1$ and $pick_1$ properly.
        
    \paragraph{Inductive step, $i > 1$.}
        For $i = 2 \ldots \Delta$, the loop at line~\ref{line:for1-no-sparsity} sets $best_i$ and $pick_i$ as required. That is, either $\tc_i$ is the largest among $\{\tc_1, \ldots, \tc_i\}$, or the largest element is among $\{\tc_1, \ldots, \tc_{i - 1}\}$ which is properly set in $best_{i - 1}$ and $pick_{i - 1}$. If there is no element with positive value, then as $best_i = 0$ and $pick_i = -1$ as desired.
        
        Similarly, for $i > \Delta$ either we choose $\tc_i$ and obtain the remaining of the output over $\{\tc_1, \ldots, \tc_{i - \Delta}\}$, or we take the best solution over $\{\tc_1, \ldots, \tc_{i - 1}\}$ without including $\tc_i$.
        
    \paragraph{Correctness of the output vector $r$.}
        Since $best_i$ and $pick_i$ are set as described above, the way $r$ is obtained trivially satisfies the constraints of $\LPdeltaprimalnok(\lambda)$. Furthermore, as $best_d$ maximizes the objective of $\LPdeltaprimalnok(\lambda)$, $r$ maximizes the objective of $\LPdeltaprimalnok(\lambda)$ as well.
    
    \paragraph{Running time.}
    Line~\ref{line:init-no-sparsity} takes $O(d)$ time. The total number of iterations of loop at line~\ref{line:for1-no-sparsity} and loop at line~\ref{line:for2-no-sparsity} is $O(d)$. Every iterations takes $O(1)$ time. The loop at line~\ref{line:while-no-sparsity} starts with $i$ being $d$, and decreases $i$ at every iteration (under natural assumption $\Delta \ge 1$). The fact that $i$ gets decreases at every iteration comes from the property that $pick_i \le i$, for every $i$. Therefore, the while loop takes $O(d)$ iterations as well, while every iteration taking $O(1)$ time. This completes the analysis.
\end{proof}

\section{Omitted proofs from Section~\ref{subsec:active-constraints}}
\label{app:active-constraints}
In this section we prove some properties of active constraints. We begin by introducing some notation. In what follows, we will be interested in cost vectors and sparsity parameters other than $c$ and $k$, respectively.
Hence, whenever this is the case, we will denote the corresponding dual LP by $\LPdeltadual(c', k')$.

We next define active constraints algorithmicaly.

\begin{algorithm}[H]
	\caption{\activeconstraints:}
	\label{alg:active-constraints}
	{\bfseries Input:} $\mc \in \bbN^{\mn}$, $w \in \bbR^{\mn + 1}$ a feasible vector of $\LPdeltadual(\mc, \mk)$ \\
	{\bfseries Output:} active constraints of $\LPdeltadual(\mc, \mk)$ for $w$
	
	\begin{algorithmic}[1]
		\STATE $\mathit{active} \gets \emptyset$; $\quad \mathit{sum}_{\Delta} \gets 0$; $\quad \mathit{last\_active} \gets -\infty$
		\FOR{$i := 1 \ldots \mn$}
			\IIF{$i - \Delta \ge 1$}{$\mathit{sum}_{\Delta} \gets \mathit{sum}_{\Delta} - w_{i - \Delta}$}
			\STATE $\mathit{sum}_{\Delta} \gets \mathit{sum}_{\Delta} + w_i$
			\IF{$\mathit{last\_active} \le i - \Delta$ and $\mc_i - \left( w_0 + \mathit{sum}_{\Delta} \right) = 0$}
				\STATE $\mathit{active} \gets \mathit{active} \cup \{i\}$
				\STATE $\mathit{last\_active} \gets i$
			\ENDIF
        \ENDFOR
		\RETURN $active$
	\end{algorithmic}
\end{algorithm}

Active constraints provide insight into the structure of $\LPdeltadual(\mc, \mk)$ that we can leverage to optimally distribute our sparsity budget over the recursive subproblems. The optimality condition of the dual program can be described in the language of active constraints as follows.
\begin{restatable}{lemma}{lemmaactiveoptimality}
\label{lemma:active-sets-optimality}
    Let $w^1 \gets \dualgreedy(\mc, v)$, $active^1 \gets \activeconstraints(\mc, w^1)$, $w^2 \gets \dualgreedy(\mc, v + 1)$, and $active^2 \gets \activeconstraints(\mc, w^2)$, for some integer $v$.
    Then, $w^1$ is an optimal solution of $\LPdeltadual(\mc, \mk)$ iff $|active^1| \ge \mk \ge |active^2|$.
\end{restatable}
\begin{proof}
We first prove two properties of actives constraints. First, observe that from the way algorithm $\activeconstraints(\mc, w)$ outputs $active$, it corresponds to tight constraints of $\LPdeltadual(\mc, \mk)$ for a given $w$. More precisely, every such tight constraint either is in $active$, or there is another tight constraint in $active$ which is at most $\Delta$ "far to the left". Furthermore, as $\mc$ is integral, it is easy to see that tight constraints for $w_0 = v$, for some integer $v$, and for $w_0 = v - \delta$, for $\delta \in [0, 1)$, are the same. Hence, active constraints for $w_0 = v$ and $w_0 = v - \delta$ are also the same. Putting these observations together, we get the following claim.
\begin{lemma}\label{lemma:active-for-delta}
    Let $v$ be an integer and $\delta \in [0, 1)$. Also, let $\hw \gets \dualgreedy(\mc, v)$ and $\hw' \gets \dualgreedy(\mc, v - \delta)$. Then,
    \[
        \activeconstraints(\mc, \hw) = \activeconstraints(\mc, \hw').
    \]
\end{lemma}
We point out that one can show even stronger statement about tight constraints, not necessarily active tough. Namely, it holds that if a constraint $i$ is tight for $\hw$ being $\dualgreedy(\mc, v)$, then it is tight for any $\dualgreedy(\mc, v')$ such that $v' \le v$. It follows from the property that for any value $v'$ there is at most one active constraint $j$ in $\activeconstraints(\mc, \dualgreedy(\mc, v'))$ such that $0 \le i - j \le \Delta$. Therefore, if $w_0 = v'$ gets decreased by "a very small" $\delta$, then the variable, e.g. $w_j$, corresponding to active constraint will decrease by $\delta$ as well. Which in turn results $i$ still being a tight constraint. Hence, the following lemma holds, which we utilize in the sequel.
\begin{lemma}\label{lemma:tight-constraints}
    Let $\hw \gets \dualgreedy(\mc, v)$ and $\tw \gets \dualgreedy(\mc, v')$, for $v' \le v$. Then, if a constraint $i$ is tight with respect to $\hw$, it is tight with respect to $\tw$ as well.
\end{lemma}

Using Lemma~\ref{lemma:active-for-delta} we can show how $\dualgreedy(\mc, v')$ changes for $v' \in (v - 1, v]$.
\begin{lemma}\label{lemma:y-delta-change}
    Let $v$ be an integer and $\delta \in [0, 1)$. By $\hw$ denote the output of $\dualgreedy(\mc, v)$, and by $\hw^\delta$ the output of $\dualgreedy(\mc, v - \delta)$. Let $active$ be returned by $\activeconstraints(\mc, \hw)$. Then
    \[
        \mk \hw^\delta_0 + \sum_{i \ge 1}\hw^\delta_i = \mk \hw_0 + \sum_{i \ge 1}\hw_i + \delta (|active| - \mk).
    \]
\end{lemma}
\begin{proof}
    From Lemma~\ref{lemma:active-for-delta} we have that $\activeconstraints(\mc, \dualgreedy(\mc, v - \delta)) = active$ for \emph{all} $\delta \in [0, 1)$. Now, by the construction of $active$, it holds
    \begin{eqnarray*}
        \mk \hw^\delta_0 + \sum_{i \ge 1}\hw^\delta_i & = & \mk (\hw_0 - \delta) + \sum_{i \in active} (\hw_i + \delta) + \sum_{i \ge 1 \text{ and } i \notin active} \hw_i \\
            & = & \mk \hw_0 + \sum_{i \ge 1} \hw_i + \delta (|active| - \mk),
    \end{eqnarray*}
    as desired.
\end{proof}

We are now ready to finalize the proof of the lemma. Let us break the equivalence stated in the lemma into two implications, and show they are true.

\paragraph{$(\impliedby)$} Let $|active^1| \ge \mk \ge |active^2|$ be true. By Lemma~\ref{lemma:y-delta-change} and the choice of $v$, we have that $\LPdeltadual_z(\mc, \mk)$ is non-decreasing for  $z \in \left[v, v + \tfrac{1}{2}\right]$ and non-increasing for $z \in [v - \tfrac{1}{2}, v]$. As $\LPdeltadual_z(\mc, \mk)$ is convex, we have $\LPdeltadual_z(\mc, \mk)$ is minimized for $z := v$.\footnote{We use the fact that from the convexity of $\LPdeltadual_z(\mc, \mk)$ we have that $\LPdeltadual_z(\mc, \mk)$ is continuous.}

\paragraph{$(\implies)$} Let $w^1$ is an optimal solution of $\LPdeltadual(\mc, \mk)$. Recall that $\LPdeltadual_z(\mc, \mk)$ is a convex function in $z$. Then, as $w^1$ is an optimum of $\LPdeltadual(\mc, \mk)$, $\LPdeltadual_x(\mc, \mk)$ is non-increasing in $z$ on interval $(-\infty, v]$ and non-decreasing on $[v, \infty)$. But then from Lemma~\ref{lemma:y-delta-change} we conclude that it can only happen if $|active^1| \ge \mk \ge |active^2|$.

\end{proof}
\section{Omitted proofs from Section~\ref{sec:perturbation-and-sparsity}}
\label{app:perturbation-and-sparsity}
In this section we finalize the proof of the correctness of our randomized algorithm. Before we delve into details, we introduce some notation. In what follows, we will be interested in cost vectors and sparsity parameters other than $c$ and $k$, respectively.
Hence, whenever this is the case, we will write $\LPdeltaprimal(c', k')$ to refer to the program $\LPdeltaprimal$ for cost vector $c'$ and sparsity $k'$. Similarly, whenever we consider some different cost vector $c'$ and sparsity parameter $k'$, we will denote the corresponding dual LP by $\LPdeltadual(c', k')$.
\\
Also, as pointed out in other sections, our running-time results are given with respect to $\gamma$, where $\gamma$ is the maximal number of bits needed to store any $c_i$. As $\LPdeltaprimal$ and $\LPdeltadual$ are invariant under shifting and multiplication of $c$, for the sake of clarity of our exposition and without loss of generality we assume $c$ is an integral vector with non-negative entries.

\lemmasingleiteration*

\subsection{A proof of Lemma~\ref{lemma:a-single-iteration}}
As pointed out already, without loss of generality in this proof we assume that $c$ is an integral non-negative vector. We also recall that we showed equivalence between problem~\eqref{eq:proj} and $\LPdeltaprimal$, and also between $\LPdeltaprimalnok$ and $\projlagr$, so in this proof we work with the LP formulations.

\paragraph{The choice of $\lambda$ is optimal.} Consider $\tw$ as in Lemma~\ref{lemma:from-dual-to-greedy-on-primal}. First we want to show that if such $\tw$ exists, then $\lambda$ obtained at line~\ref{line:main-lambda} of Algorithm~\ref{alg:main} is such that it also defines $k$ active constraints. This in turn would imply, by Lemma~\ref{lemma:from-dual-to-greedy-on-primal}, that support $\hS$ obtained at line~\ref{line:optimal-hc} has cardinality $k$, and hence is an optimal solution to $\LPdeltaprimal(\tc)$.

If $\tw_0 = \lambda$, then we are done. Otherwise, assume that $\tw_0 \neq \lambda$. By Lemma~\ref{lemma:active-sets-optimality}, $\tw$ is an optimal solution of $\LPdeltadual(\tc)$. On the other hand, as we discussed in Section~\ref{sec:duality-and-back}, $\lambda$ is such that $\val{\LPdeltadual_{\lambda}(\tc)} = \val{\LPdeltadual(\tc)}$, and $\tw_0 < \lambda$ by the choice of $\lambda$. Therefore, by Lemma~\ref{lemma:active-sets-optimality} it holds that $\dualgreedy(\tc, \lambda)$ defines at least $k$ active constraints. Furthermore, as $\LPdeltadual(\tc)$ is convex w.r.t. to the variable $w_0$, then for every $\lambda' \in [\tw_0, \lambda]$ we have $\val{\LPdeltadual_{\lambda'}(\tc)} = \val{\LPdeltadual(\tc)}$. In other words, $\val{\LPdeltadual_{\lambda'}(\tc)}$ remains constant over the given interval. Hence, from Lemma~\ref{lemma:y-delta-change} we conclude that $\dualgreedy(\tc, \lambda)$ defines exactly $k$ active constraints.

However, for $\tc$ given on the input there might not exists any $\lambda$ such that $\dualgreedy(\tc, \lambda)$ defines exactly $k$ active constraints. Our goal is to show that the randomization we apply assures that for the obtained $\tc$ it is always the case that there is some $\lambda$ so that $\dualgreedy(\tc, \lambda)$ has exactly $k$ active constraints.

\paragraph{The evolution of active constraints.} Now we want to show that the randomization we apply will result in an existence of $\tw$ as described in Lemma~\ref{lemma:from-dual-to-greedy-on-primal}. We start by studying the evolution of active constraints defined by the output of $\dualgreedy(\tc, w_0)$ as $w_0$ decreases.

First, recall that by Lemma~\ref{lemma:tight-constraints} we have that if a constraint becomes tight with respect to some $\hw \gets \dualgreedy(\tc, \lambda)$, it remains tight with respect to $\hw' \gets \dualgreedy(\tc, \lambda')$ for every $\lambda' \le \lambda$. Let $T$ denote the set of tight constraints with respect to $\hw$, and $T'$ with respect to $\hw'$. By our discussion $T \subseteq T'$.
\\
Let $A$ and $A'$ be the set of active constraints with respect to $\dualgreedy(\tc, \lambda)$ and $\dualgreedy(\tc, \lambda')$, respectively. Clearly $A \subseteq T$ and $A' \subseteq T'$. We claim that $|T \cap A'| \le |A|$. Observe that $A$ is a minimum set of constraints so that every tight constraint is covered (covered in the natural way). However, $A$ is also a maximum set of constraints of $T$ that can be chosen so that no two of them overlap, i.e. so that every two of them are at least $\Delta$ apart. That means if one would choose a subset of $T$ larger than $|A|$ then some two constraints would overlap. Hence, such a subset can not consist of only active constraints, and therefore $|T \cap A'| \le |A|$.

This implies that if the number of active constraints increases by $a > 0$ at certain point, then there are at least some $a$ constraints that became active, but also tight, for the first time. Now we want to study what is the probability that two or more non-tight constraints become tight with respect to $\dualgreedy(\tc, \lambda)$, for any $\lambda$.

Let us focus on a single constraint $j$. Fix randomness of all the $X_i$ for $i \neq j$, i.e. fix $\tc_i$ for all $i \neq j$. For $X_j = 0$, there are at most $d - 1$ different values of $w_0$ when any of those constraints becomes tight for the first time. Let $W$ denote the set of these $w_0$ values. So $|W| < d$. Now, construct set $W_j$ as follows. For each $\lambda \in W$:
\begin{itemize}
    \item If constraint $j$ is already tight with respect to $\dualgreedy(\tc, \lambda)$ for $X_j = 0$, do nothing.
    \item Define $\hc_i = \tc_i$ for $i \neq j$, and $\hc_j = \tc_j + x_{\lambda}$, where $x_{\lambda}$ is defined as the least value so that constraint $j$ becomes tight for the first time with respect to $\dualgreedy(\hc, \lambda)$. Observe that $x_{\lambda} \ge 0$ and $x_{\lambda}$ is integral. If $x_{\lambda} < d^3$, add $x_{\lambda}$ to $W_j$.
\end{itemize}
We have $|W_j| \le |W| < d$. Each of the value of $W_j$ correspond to some value of $X_j$. Also, observe that for given $\lambda$, a constraints can become tight for the first time for at most one value of $X_j$. In addition, as long as constraint $j$ is not tight it does not affect when the other constraints will become tight, as $w_j = 0$ so $w_j$ has no affect on other constraints. This all implies that there are at most $d-1$ distinct values of $X_j$, out of $d^3$ of them, when constraint $j$ and some other constraint become tight. Therefore,
\[
    \prob{\text{constraint $j$ and any other constraint become tight for the same value of $w_0 = \lambda$}} < \frac{d}{d^3} = \frac{1}{d^2}.
\]
Now we can apply union bound to conclude
\[
    \prob{\text{two constraints become tight for the same value of $w_0 = \lambda$}} < d \frac{1}{d^2} = \frac{1}{d}.
\]
Therefore, after applying randomness, for every value of $w_0 = \lambda$ at most one constraint becomes tight with probability at least $1 - 1/d$. Following our discussion above, this in turn implies that the number of active constraints increases by at most 1 after decreasing the value of $w_0$ by 1. Therefore, there is $\tw$ as in Lemma~\ref{lemma:from-dual-to-greedy-on-primal} with probability $1 - 1/d$ at least.

\paragraph{Required randomness.}
Every $c_i$ we perturb by one out of $d^3$ different values, for which $O(\log{d})$ random bits suffices. Therefore, in total we need $O(d \log{d})$ random bits.



This concludes the proof.
\section{Deterministic worst-case nearly-linear time algorithm}
\label{section:deterministic-algorithm}
In this section we describe our deterministic algorithm. For the sake of clarity, we repeat some of the content presented in earlier sections.

Our algorithm stems from a linear programming view on the separated sparsity recovery. It has already been shown that this LP is totally unimodular \cite{hegdeLP}, which implies that solving the LP provides an integral solution and hence solves separated sparsity recovery. However, the previous work resorts to a black-box approach for solving this LP leading to a prohibitive $O(d^{3.5})$ time complexity. We make a step forward, and study the dual LP, obtaining a method that takes nearly-linear time to find its optimal solution. By the strong duality, the value of the dual is the value of the primal LP as well. However, unfortunately, it is not hard to see that a generic relation between primal and dual LP solutions known as "complementary slackness" does not lead to recovery of the primal optimal solution itself.

To cope with that shortcoming, we further analyze the properties of dual solution. We exhibit very close connection between its structure and the sparsity of the primal solution, which we present via the notion of "active constraints". Intuitively, this structure allows us to characterize the cases in which the complementary slackness in fact provides the primal from a dual optimal solution. Then we use these findings in our algorithm to slightly perturb the input instance, while not affecting the value of the solution, so that it is possible to obtain a primal from a dual solution in the general case, and hence solve the separated sparsity recovery.

In Section~\ref{sec:dual} we defined LP $\LPdeltaprimal$, which corresponds to the problem~\eqref{eq:proj-massaged}, as follows:
\begin{alignat*}{4}
	\text{maximize}   &\qquad & c^T u \\
	\text{subject to} &\qquad & \sum_{i = 1}^d{u_i} & = k & \qquad & \\
					&\qquad & \sum_{j = i}^{\min\{i + \Delta - 1, n\}}{u_j} & \le 1 & \qquad & \forall i = 1 \ldots d \\
					&\qquad & u_i  & \ge 0 &\qquad & \forall i = 1 \ldots d
\end{alignat*}
We also defined the dual LP to $\LPdeltaprimal$, denoted by $\LPdeltadual$, as
\begin{alignat*}{4}
	\text{minimize}   &\qquad & w_0 k + \sum_{i = 1}^{d} w_i \\
	\text{subject to} &\qquad & w_0 + \sum_{\substack{j\ :\ j \ge 1 \text{ and } \\ j \le i \le j + \Delta - 1}}{w_j} & \ge c_i & \qquad & \forall i = 1 \ldots d \\
					&\qquad & w_i  & \ge 0 &\qquad & \forall i = 1 \ldots d \\
					&\qquad & w_0  & \in \mathbb{R} &\qquad &
\end{alignat*}
As Theorem~\ref{theorem:ternary-search}, which proofs appears in Section~\ref{app:dual}, this dual LP has a combinatorial structure that enables us to solve it in nearly-linear time.

Before delving into details, we introduce some notation. In what follows, we will be interested in cost vectors and sparsity parameters other than $c$ and $k$, respectively.
Hence, whenever this is the case, we will write $\LPdeltaprimal(c', k')$ to refer to the program $\LPdeltaprimal$ for cost vector $c'$ and sparsity $k'$. Similarly, whenever we consider some different cost vector $c'$ and sparsity parameter $k'$, we will denote the corresponding dual LP by $\LPdeltadual(c', k')$.
\\
As already mentioned, our running-time results are given with respect to $\gamma$, where $\gamma$ is the maximal number of bits needed to store any $c_i$. As $\LPdeltaprimal$ and $\LPdeltadual$ are invariant under shifting and multiplication of $c$, for the sake of clarity of our exposition and without loss of generality we assume $c$ is an integral vector with non-negative entries.

At a high-level, we develop our algorithm in three main steps. In Section \ref{sec:dual} we considered the dual of the LP $\LPdeltaprimal$ and show that its combinatorial structure can be exploited to compute an optimal solution to it in nearly-linear time. By strong duality, this optimal dual solution gives us then the \emph{value} of the optimal primal solution. But, unfortunately, it does not give us the optimal primal solution itself.

To alleviate this issue, we develop a divide-and-conquer procedure for extracting that optimal primal solution. First, in Section~\ref{section:recover-delta}, we demonstrate that by analyzing answers of the dual oracle on perturbed versions of the original problem we can quickly recover a single non-zero entry of the optimal primal solution. That entry can be used to partition our instance into two smaller subproblems. Then, in Section~\ref{section:distributing-sparsity}, we show how to make these two subproblems fully independent by devising an optimal split of the sparsity constraint that they share. This optimal split is extracted from the structure of the dual solutions.

With these components in place, we assemble our final algorithm in Section~\ref{sec:recoverprimal}.

\subsection{Recovering a segment of an optimal solution}
\label{section:recover-delta}
At this point, we developed a way of computing the optimal value $OPT$ of $\LPdeltaprimal$.
However, this is not sufficient for our purposes as the separated sparsity problem also requires us to provide the corresponding solution, i.e., a binary vector $\ustar$ corresponding to $OPT$.

One might hope that this primal solution $\ustar$ can be inferred from the dual solution $\hw$ that our algorithm $\optimalvalue(\mc, \mk)$ (see Algorithm \ref{alg:dual}) provides. It is not hard to see, however, that the generic relationship between the optimal primal and optimal dual solutions that the so-called ``complementary slackness'' provides is not sufficient here. 

Therefore, we instead design a problem-specific algorithm for finding the desired primal solution vector $\ustar$. As a first step, we focus on the task of recovering a single segment of $\ustar$. Specifically, given some target segment $S$ of at most $\Delta$ consecutive entries, we want to either find a \emph{single} index $i\in S$ such that $\ustar_i = 1$, for some optimal primal solution $\ustar$; or to conclude that $\ustar_i = 0$ for all $i\in S$.

To recover such a segment, we define in an adaptive manner a family of cost vectors $c^1, \ldots, c^t$, for some $t \in O(\log{\Delta})$  and invoke our dual LP solver $\optimalvalue(\mc, \mk)$ on these cost vectors.
As we show, the solutions to these perturbed instances allow us to infer $\ustar_j$ for all $j \in S$.

To provide more details, let us fix some $\Delta$-length segment $S=[j_s, j_e]$, i.e., $j_e = \min\{j_s + \Delta - 1, n\}$. We now want to decide whether there is an optimal solution $\hu$ to $\LPdeltaprimal$ such that $\hu_i = 1$ for some $i \in [j_s, j_e]$. Observe that there might be another optimal solution $\tu$ to $\LPdeltaprimal$. Furthermore, it might be the case that $\tu_j = 0$ for every index $j \in [j_s, j_e]$, while there is an index $i \in [j_s, j_e]$ such that $\hu_i = 1$.
So while designing the algorithm, we distinguish two cases.
First, we analyse the case in which there exists an optimal solution $\hu$ to $\LPdeltaprimal$ so that $\hu_i = 0$ for every $i \in [j_s, j_e]$.
Next, we consider the complementary case in which for every optimal solution $\hu$, we have that $\hu_i = 1$ for some $i \in [j_s, j_e]$. In the latter case, we recover an index $j$ such that there is a solution $\hu$ for which $\hu_j = 1$ holds.

The first case is captured by the following claim.
\begin{restatable}{lemma}{lemmadeltarecovery}
\label{lemma:delta-recovery--1}
    Let $S \subseteq [\mn]$ be a set of indices, let $\mc \in \bbN^\mn$ be a coefficient vector, and let $\mk$ be the sparsity.
    Define a vector $c'' \in \bbN^\mn$ as follows:
    \[
        c_i'' =
        \begin{cases}
            1 & \text{if } i \in S \\
            1 + \mc_i & \text{otherwise}
        \end{cases} \; .
    \]
    Then, $\optimalvalue(c'', \mk)$ equals $\optimalvalue(\mc, \mk) + \mk$ iff there exists an optimal solution $\ustar$ to $\LPdeltaprimal(\mc, \mk)$ such that $\ustar_i = 0$ whenever $i \in S$.
\end{restatable}
\begin{proof}
    Let us show the two direction of equivalence separately.
\paragraph{$(\implies)$} Assume that $\optimalvalue(c'', \mk)$ equals $\optimalvalue(\mc, \mk) + \mk$. Let $\hu$ be an optimal solution to $\LPdeltaprimal(\mc, \mk)$ for the cost vector given by $c''$. Now we want to show that $\hu_i = 0$ for every $i \in S$. Towards a contradiction, assume it is not the case, i.e. there exists $j \in S$ such that $\hu_j = 1$. But then, $\sum_{i\ :\ \hu_i = 1} (c_i + 1) > \sum_{i\ :\ \hu_i = 1} c_i''$ as $\mc_i + 1 > c_i''$ for $i \in S$, and hence $\optimalvalue(c'', \mk) < \optimalvalue(\mc, \mk) + \mk$, contradicting our assumption.

\paragraph{$(\impliedby)$}
First, observe that $\optimalvalue(c'', \mk) \le \optimalvalue(\mc, \mk) + \mk$ as $c'' \le \mc + 1$. On the other hand, if there exists an optimal solution $\ustar$ to $\LPdeltaprimal(\mc, \mk)$ such that $\ustar_i = 0$ whenever $i \in S$, then $\ustar$ achieves value $\optimalvalue(\mc, \mk) + \mk$ in $\LPdeltaprimal(\mc, \mk)$ for the costs given by $c''$. In other words, restricted to the set of indices outside of $S$, $\mc$ is a shifted by 1 variant of $c''$. Therefore, we also have $\optimalvalue(c'', \mk) \ge \optimalvalue(\mc, \mk) + \mk$. Now this implies $\optimalvalue(c'', \mk) = \optimalvalue(\mc, \mk) + \mk$, as desired.
\end{proof}

Next we provide algorithm $\deltarecovery$ that we use in the proof of the next theorem.
\begin{algorithm}[H]
	\caption{\deltarecovery:}
	\label{alg:deltarecovery}
	Input: $\mc \in \bbN^{\mn}$, sparsity $\mk$, index $j_s$ \\
	Output: index $r$ as described in Theorem~\ref{theorem:delta-recovery}
	\vspace{-0.1in}
	\begin{enumerate}[1.]
	    \setlength\itemsep{-0.3em}
		\item\label{item:init-idx} $j_e \gets \min\{j_s + \Delta - 1, \mn\}$
		\item\label{item:init} $s \gets j_s$; $\quad e = j_e$; $\quad OPT \gets \optimalvalue(\mc, \mk)$
		\item $c_i'' \gets
            \begin{cases}
                1 & \text{if } i \in [j_s, j_e] \\
                1 + \mc_i & \text{otherwise}
            \end{cases}
            $
        \item\label{item:ignore} If $\optimalvalue(c'', \mk) = OPT + \mk$ then return -1
		\item\label{line:binary-search-loop} While $s < e$
			\begin{enumerate}[(a)]
			    \setlength\itemsep{-0.3em}
				\item $mid \gets \left \lfloor \tfrac{s + e}{2} \right \rfloor$
				\item\label{item:c''-inside-loop} $c_i'' \gets
		            \begin{cases}
		                1 & \text{if } i \in [j_s, j_e] \setminus [s, mid] \\
		                1 + \mc_i & \text{otherwise}
		            \end{cases}
		            $
				\item\label{item:update-e} If $\optimalvalue(c'', \mk) = OPT + \mk$ then $e \gets mid$
				\item\label{item:update-s} Else $s \gets mid + 1$
			\end{enumerate}
        \item Return $s$
	\end{enumerate}
\end{algorithm}

\begin{restatable}{theorem}{theoremdeltarecovery}
\label{theorem:delta-recovery}
    There exists an algorithm that given $\mc \in \bbN^\mn$, sparsity $\mk$, and an index $j_s \in [1, \mn]$, outputs an integer $r$ having the following properties:
    \begin{itemize}
        \item If for every optimal solution $\hu$ to $\LPdeltaprimal(\mc, \mk)$ there is an index $i$ such that $\hu_i = 1$ and $i \in [j_s, \min\{j_s + \Delta - 1, \mn\}]$, then $r$ is set to be an index in $[j_s, \min\{j_s + \Delta - 1, \mn\}]$ such that there is an optimal solution $\ustar$ for which we have $\ustar_r = 1$.
        \item Otherwise, $r$ is set to -1.
    \end{itemize}
    Furthermore, if the algorithm runs in time $O\left((\log{\Delta}) \rtoptimalvalue{\mn}{(\mcmax + 1)}{\mk}\right)$.
\end{restatable}
\begin{proof}
On line~\ref{item:ignore}, $\deltarecovery$ first checks whether we can simply ignore all the entries indexed by $\{i, \ldots, \min\{i + \Delta - 1, \mn\}\}$. And if yes, it returns -1. (This step is formalized in the statement of Theorem~\ref{theorem:delta-recovery}.) However, entries from those interval can only be disregarded if there exists an optimal solution $\ustar$ to $\LPdeltaprimal(\mc, \mk)$ such that for every $j \in \{i, \ldots, \min\{i + \Delta -1 , \mn\}\}$ we have $\ustar_j = 0$. So, if there is no such $\ustar$, i.e. the entries corresponding to that interval have to be considered, the lines following line~\ref{item:ignore} of $\deltarecovery$ serve to pinpoint an entry $j \in \{i, \ldots, \min\{i + \Delta - 1, \mn\}\}$ so that there exists an optimal solution $\hu$ to $\LPdeltaprimal(\mc, \mk)$ for which holds $\hu_j = 1$. The way it is done is by applying a binary search over the interval that \emph{can not} be ignored, tracked via variables $s$ and $e$ of the algorithm.

\paragraph{Correctness.} Algorithm $\deltarecovery$ outputs $-1$ correctly by Lemma~\ref{lemma:delta-recovery--1} for $S = \{j_s, \ldots, j_e\}$.

Next, we show that when the binary search loop starting at line~\ref{line:binary-search-loop} ends we have $s = e$, i.e. $[s, e]$ corresponds to a single index. As long as $\Delta \ge 1$ and $j_s \le \mn$, at line~\ref{item:init-idx} and line~\ref{item:init} values $s$ and $e$ are initialized so that $s \le e$, so initially $[s, e]$ is indeed a non-empty interval. Furthermore, as we have $s < e$ in each iteration, it holds $mid < e$. But it also holds $mid \ge s$ (in fact $mid$ equals $s$ only when $s + 1$ equals $e$). So, updating $e$ to $mid$ at line~\ref{item:update-e}, or $s$ to $mid + 1$ at line~\ref{item:update-s}, $[s, e]$ remains a non-empty interval in any iteration. Notice that the update rules also guarantee that we either increase $s$ or decrease $e$ at each iteration, and therefore $[s, e]$ shrinks its size by 1 at least at each iteration. Putting this together, and taking into account the loop-termination condition at line~\ref{line:binary-search-loop}, we conclude that after the loop ends it holds $s = e$.

We say that interval $[s, e]$ has property $R_{s, e}$ if there is an optimal solution $\hu$ to $\LPdeltaprimal(\mc, \mk)$ such that $\hu_i = 1$ for some $i \in [s, e]$. Now, if we show that after line~\ref{item:ignore} at every step of the algorithm interval $[s, e]$ has property $R_{s, e}$, then the proof will follow immediately. That is exactly how we proceed. Namely, we show that if $[s, e]$ has property $R_{s, e}$, then after updating $s$ or $e$ at line~\ref{item:update-e} or line~\ref{item:update-s} obtaining $s'$ and $e'$, respectively, then interval $[s', e']$ will have property $R_{s', e'}$.

Observe that at the beginning of the very first iteration of the loop property $R_{s, e}$, equivalent to $R_{j_s, j_e}$, holds as by Lemma~\ref{lemma:delta-recovery--1} \emph{every} optimal solution $\hu$ is such that $\hu_i = 1$ for some $i \in [j_s, j_e]$. Next, let $[s, e]$ be updated to $[s', e']$. We want to show that property $R_{s', e'}$ holds as well.

First, assume that line~\ref{item:update-e} gets executed, i.e. $s' = s$ and $e' = mid$. Then, by Lemma~\ref{lemma:delta-recovery--1}, we have that there exists an optimal solution $\hu$ to $\LPdeltaprimal(\mc, \mk)$ such that $\hu_i = 0$ for every $i \in [j_s, j_e] \setminus [s, mid]$. But we also have that there is $i \in [j_s, j_e]$ such that $\hu_i = 1$. So, putting it together, we conclude that property $R_{s', e'}$ holds.

Next, assume that the if condition at line~\ref{item:update-e} does not hold. So, line~\ref{item:update-s} is executed, i.e. $s' = mid + 1$ and $e' = e$. Then again by Lemma~\ref{lemma:delta-recovery--1}, and as a consequence of both line~\ref{item:ignore} and line~\ref{item:update-e}, we have that for every optimal solution $\hu$ we have $\hu_i = 1$ for some $\in [j_s, j_e] \setminus [s, mid]$. Also, as $j_e - j_s + 1 \ge \Delta$, we also have $\hu_j = 0$ for every $j \in [s, mid]$. But then, as by our assumption there is an optimal solution $\tu$ and an index $i$ such that $\tu_i = 1$ and $i \in [s, e]$, we have that $i \in [mid + 1, e]$. Hence property $R_{mid + 1, e}$, which is equivalent to $R_{s', e'}$, holds.

\paragraph{Running time.} If $\deltarecovery$ outputs -1 at line~\ref{item:ignore}, it invokes $\optimalvalue$ for the cost vector $\mc$ and $c''$. Note that $\cmax' \le \mcmax + 1$. Then, by Theorem~\ref{theorem:ternary-search} that case takes $O\left(\rtoptimalvalue{\mn}{(\mcmax + 1)}{\mk}\right)$ time.

If the method does not output -1, then it enters the while loop. The loop applies a standard binary search over $j_s$ and $j_e$ which are set so that it holds $j_e - j_s + 1 \le \Delta$. So, the loop iterates for $O(\log{\Delta})$ times. Every iteration invokes $\optimalvalue$ for vector $c''$ as defined at line~\ref{item:c''-inside-loop}, resulting in the total running time of $O((\log{\Delta}) \rtoptimalvalue{\mn}{(\mcmax + 1)}{\mk})$.
\end{proof}

\subsection{Distributing sparsity}
\label{section:distributing-sparsity}
The algorithm we presented in the previous section enables us to quickly recover a target length-$\Delta$ segment of some optimal primal solution $\ustar$. We now would like to build on this procedure to develop a divide-and-conquer approach to recovering the primal solution $\ustar$ in full. 

Our intention is to use the procedure from Section \ref{section:recover-delta} to split our input problem into two (smaller and approximately equally sized) subproblems by recovering a ``middle'' segment of the solution $\ustar$ and then to proceed recursively on each of these subproblems. The difficulty here, however, is that these two resulting subproblems are not really independent, even though they correspond to separate cost vectors.
The subproblems still share the sparsity constraint.
That is, to make these subproblems truly independent, we must also specify the split of our sparsity ``budget'' $\mk$ between them.

In this section, we analyze properties of the dual LP and exhibit very close ties between its structure, as captured by active constraints, and the optimal sparsity distribution for our two subproblems. Specifically, we establish the following theorem.
\begin{restatable}{theorem}{theoremdistribute}
\label{theorem:distribute-sparsity}
	Let $\mc \in \bbN^{\mn}$, let $k''$ be a sparsity parameter, and let indices $s$ and $e$ be such that $e - s + 1 \ge \Delta$ or $e$ equals $\mn$.
	Define a vector $c''\in \bbN^{\mn}$ as follows:
	\[
	    c''_i =
	    \begin{cases}
	        -\infty & \text{ if } i \in [s, e] \\
	        \mc_i & \text{otherwise}
	    \end{cases}
	\]
	Assume that there is a $\Delta$-separated choice of coordinates of $c''$ such that $k''$ coordinates are chosen, and none of them has an index in $[s, e]$, i.e., the instance is feasible when restricted to coordinates outside of $[s, e]$.
	Then, there is an algorithm that outputs $k_L'$ and $k_R'$ with the property that there exists an optimal $\Delta$-separated solution for cost vector $c''$ and sparsity $k''$ such that it chooses $k_L'$ coordinates of $c''$ with indices less than $s$, and $k_R'$ coordinates of $c''$ with indices greater than $e$.
	
	The algorithm runs in time $O\left(\mn (\log{\mcmax} + \log{k''}) \right)$.
\end{restatable}
Observe that Theorem~\ref{theorem:distribute-sparsity} essentially gives an algorithm that distributes sparsity in an optimal way and thus enables us to implement the desired divide-and-conquer approach.

Let us illustrate how to use the algorithm $\deltarecovery$ (Algorithm \ref{alg:deltarecovery}) together with $\activeconstraints$ to recover an optimal solution for the example $c = (3, 2, 1, 4, 1, 1, 2, 1)$, $k = 4$ and $\Delta = 2$.
Let $\ustar$ be an optimal solution to $\LPdeltaprimal$.
As described above, we proceed by splitting $c$ into two subvectors and solve the separated sparsity problem on each of them independently.
The first step is to find the exact splitting point in $c$.
To achieve this, we invoke $\deltarecovery(c, 4, k)$.
As $c_4 = 4$ is part of any optimal solution, the method returns index $4$ and hence we have $\ustar_4 = 1$.
So at this moment we have learned one index of the optimal solution.
In fact, it also implies that $\ustar_3 = \ustar_5 = 0$.
Hence it remains to learn the remaining $k - 1$ indices that define $\ustar$.
However, we have to distribute the remaining sparsity $k-1$ over the two subproblems.
To that end, we define a new vector $c'$ as follows
\[
    c' =
    \begin{cases}
        -\infty & \text{if } i \in \{3, 4, 5\} \\
        c_i & \text{otherwise.}
    \end{cases}
\]
The definition of $c'$ enforces that no optimal solution to the separated sparsity problem for sparsity $k-1$ will choose $c_3$, $c_4$, or $c_5$ (setting $c'_i$ to the dummy value $-\infty$ is for illustration purposes only).
Next, we invoke $\activeconstraints(c', \optimalvalue(c', k - 1))$ and denote its output by $active'$. 
Observe that, as long as $\optimalvalue(c, k)$ has finite value, an optimal solution to $\optimalvalue(c', k - 1)$ is finite as well.
We have that $active' = \{1, 6, 8\}$.
Now we count the number of elements in $active'$ that are ``to the left'' and ``to the right'' of $\ustar_4$, denoting these quantities by $\hk$ and $\tk$ respectively, i.e.,
\[
    \hk = |\{i \in active'\ |\ i < 4\}|, \quad \text{ and } \quad \tk = |\{i \in active'\ |\ i > 4\}| \; .
\]
In a similar way, we split the vector $c'$ into $\hc$ and $\tc$ as follows:
\[
    \hc = (c_1, c_2), \text{ and } \tc = (c_6, c_7, c_8) \; .
\]
Finally, we solve two $\Delta$-separated instances independently, one for $\hc$ and $\hk$, and the second one for $\tc$ and $\tk$.

\begin{proofof}{Theorem~\ref{theorem:distribute-sparsity}}

Let us start by presenting an algorithm the distributes the sparsity, and continue by analyzing it.
\begin{algorithm}[H]
	\caption{\distributesparsity:}
	Input: $\mc \in \bbN^{\mn}$, sparsity $k''$ to be distributed, indices $s$ and $e$ as described in Theorem~\ref{theorem:distribute-sparsity} \\
	Output: sparsity $k_L'$ for the left side, sparsity $k_R'$ for the right side
	\vspace{-0.1in}
	\begin{enumerate}[1.]
	    \setlength\itemsep{-0.3em}
		\item $shift \gets 1 + (k'' \mcmax + 1)$
		\item $c_i'' \gets
		        \begin{cases}
		            1 & \text{if } i \in [s, e] \\
		            \mc_i + shift & \text{otherwise}
		        \end{cases}
		    $
        \item\label{item:w1-opt-int} $w^1 \gets \optimalvalue(c'', k'', -(k'' - 1) \mcmax + shift, \mcmax + shift)$
        \item $active^1 \gets \activeconstraints(c'', w^1)$
        \item $active^2 \gets \activeconstraints(c'', \dualgreedy(c'', w^1_0 + 1))$
	    \item $k^1_L \gets |\{i \in active^1\ |\ i < s\}|$
	    \item $k^2_L \gets |\{i \in active^2\ |\ i < s\}|$
	    \item $k^2_R \gets |\{i \in active^2\ |\ i > e\}|$
	    \item $k_L' \gets k^2_L + \min\{k^1_L - k^2_L, k'' - k^2_L - k^2_R\}$
	    \item $k_R' \gets k'' - k_L'$
	    \item Return $(k_L', k_R')$
	\end{enumerate}
\end{algorithm}

\paragraph{Optimality for $\LPdeltadual(c'', k'')$.} We argue that $w^1$ defined on line~\ref{item:w1-opt-int} is an optimal solution to $\LPdeltadual(c'', k'')$ although $w^1_0$ is restricted to belong to interval $[-(k'' - 1) \mcmax + shift, \mcmax + shift] = [\mcmax + 2, \mcmax + shift]$. First, observe that for $w^1_0 \ge \mcmax + 2$ no constraint with index in $[s, e]$ is tight, and hence no such constraint can be active. Furthermore, as no constraints in $[s, e]$ is tight, we have $w^1_i = 0$ for all $i \in [e, s]$. Now, consider $c_L = (\mc_1, \ldots, \mc_{s - 1})$ and $c_R = (\mc_{e + 1}, \ldots, \mc_\mn)$ as defined in the algorithm. Then, by following the proof of Lemma~\ref{lemma:ystar-lower-bound}, and as $e - s + 1 \ge \Delta$, we conclude that $\dualgreedy(c'', \mcmax + 2)$ has at least $\min\left\{\mk, \left\lceil \tfrac{s - 1}{\Delta} \right \rceil \right\} + \min\left\{\mk, \left\lceil \tfrac{n - e}{\Delta} \right \rceil \right\}$ active constraints which, by the construction and the assumption that the input instance is feasible, is $k''$ at least. So, by the monotonicity of the size of active constraints given by Lemma~\ref{lemma:active-constraints-monotone} and the optimality condition provided via Lemma~\ref{lemma:active-sets-optimality} we have that $w^1$ is an optimal solution to $\LPdeltadual(c'', k'')$.

\paragraph{Correctness of distributed sparsity.} Let $active^1$ denote $\activeconstraints(c'', w^1)$ and $active^2$ denote $\activeconstraints(c'', w^2)$, where $w^2 \gets \dualgreedy(c'', w^1_0 + 1)$. Next, split $w^1$ and $w^2$ as follows. Let $w^1_L$ be a zero-indexed $s$-dimensional and $w_R^1$ be a zero-indexed $(\mn - e + 1)$-dimensional vector defined as
\[
    w^1_L \gets (w^1_0, w^1_1, \ldots, w^1_{s - 1}), \text{ and } w^1_R \gets (w^1_0, w^1_{e + 1}, \ldots, w^1_{\mn}).
\]
Intuitively, letter 'L' stands for left to $s$ and letter 'R' stands for right to $e$. Similarly to $w^1_L$ and $w^1_R$, define $w^2_L$ and $w^2_R$ as
\[
    w^2_L \gets (w^2_0, w^2_1, \ldots, w^2_{s - 1}), \text{ and } w^2_R \gets (w^2_0, w^2_{e + 1}, \ldots, w^2_{\mn}).
\]
Also, split $active^1$ and $active^2$ with respect to $s$ and $e$ in the obvious way
\[
    active^1_L \gets \{i \in active^1\ |\ i < s\}, \text{ and } active^1_R \gets \{i \in active^1\ |\ i > e\}, 
\]
and
\[
    active^2_L \gets \{i \in active^2\ |\ i < s\}, \text{ and } active^2_R \gets \{i \in active^2\ |\ i > e\}.
\]
As a reminder, there is no $i \in [s, e]$ such that $i \in active^1$. Furthermore, as there is no tight constraint in $[s, e]$ for $w^1$, we have $w^1_i = 0$ for every $i \in [s, e]$. Observe that the same holds for $active^2$ and $w^2$. In addition, we have $e - s + 1 \ge \Delta$. From this, we can derive the following list of equalities, which essentially allows us to split the input problem into two independent subproblems. So, we have: $\dualgreedy(c_L, w^1_0)$ equals $w^1_L$; $\dualgreedy(c_L, w^1_0 + 1)$ equals $w^2_L$; $\dualgreedy(c_R, w^1_0)$ equals $w^1_R$; and, $\dualgreedy(c_R, w^1_0 + 1)$ equals $w^2_R$. But also, we have: $\activeconstraints(c_L, w^1_L)$ equals $active^1_L$; $\activeconstraints(c_L, w^2_L)$ equals $active^2_L$; $\activeconstraints(c_R, w^1_R)$ equals $active^1_R$; and, $\activeconstraints(c_R, w^2_R)$ equals $active^2_R$.

Now, by Lemma~\ref{lemma:active-constraints-monotone} we conclude that $|active^1_L| \ge |active^2_L|$ and $|active^1_R| \ge |active^2_R|$. This in turn implies that $k_L'$ and $k_R'$ in $\distributesparsity$ are derived correctly. (Note that $k_L' + k_R' = k''$.) But then, by Lemma~\ref{lemma:active-sets-optimality} we have that $w^1_L$ is an optimal solution to $\LPdeltadual(c_L, k_L')$ and $w^1_R$ is an optimal solution to $\LPdeltadual(c_R, k_R')$, so we can independently solve $\recoverprimal(c_L, k_L')$ and $\recoverprimal(c_R, k_R')$ knowing that the optimal solution, i.e. its value, remains the same. We point out that $w_L' \gets \optimalvalue(c_L, k_L)$ might differ from $w^1_L$, and similarly $w_R' \gets \optimalvalue(c_R, k_R)$ might differ from $w^1_R$. In fact, $w_L'$ and $w_R'$ might be such that $w_L'{}_0 \neq w_R'{}_0$, although we have that $w^1_L{}_0$ is equal to $w^1_R{}_0$.

\paragraph{Running time.} Every line, except maybe line~\ref{item:w1-opt-int}, take $O(\mn)$ time. On the other hand, line~\ref{item:w1-opt-int} takes time $O\left(\mn \log{\{\mcmax + shift - (-(k'' - 1)\mcmax + shift)\}} \right)$ which is equal to $O\left(\mn (\log{\mcmax} + \log{k''}) \right)$.
\end{proofof}

\subsection{Separated sparsity in nearly-linear time}
\label{sec:recoverprimal}
We now have all components in place to state our final algorithm $\recoverprimal(\mc, \mk)$ that produces an optimal integral solution for $\LPdeltaprimal(\mc, \mk)$ (see Algorithm \ref{alg:recoverprimal}).
\begin{algorithm}
	\caption{\recoverprimal:}
	\label{alg:recoverprimal}
	Input: $\mc \in \bbN^{\mn}$, sparsity $\mk$ \\
	Output: a primal solution $\mn$-dimensional vector $\ustar$
	\vspace{-0.1in}
	\begin{enumerate}[1.]
	    \setlength\itemsep{-0.3em}
	    \item\label{line:k=0} If $\mk = 0$ return $\allzeros$
		\item\label{item:get-r} $r \gets \deltarecovery(\mc, \left \lfloor \tfrac{\mn + 1}{2} \right \rfloor, \mk)$
		\item If $r = -1$ then
		    \vspace{-0.1in}
		    \begin{enumerate}[1.]
		        \setlength\itemsep{-0.3em}
		        \item $s \gets \left \lfloor \tfrac{\mn + 1}{2} \right \rfloor$, $\quad e \gets \min\{s + \Delta - 1, \mn\}$
		        \item $k'' \gets \mk$
		    \end{enumerate}
		\item Else
		    \vspace{-0.1in}
		    \begin{enumerate}[1.]
                \setlength\itemsep{-0.3em}
                \item $s \gets \max\{r - \Delta + 1, 1\}$, $\quad e \gets \min\{r + \Delta - 1, \mn\} $
                \item $k'' \gets \mk - 1$
            \end{enumerate}
        \item\label{item:distribute-sparsity} $(k_L', k_R') \gets \distributesparsity(\mc, k'', s, e)$
        \item\label{item:split-c} $c_L \gets (\mc_1, \ldots, \mc_{s - 1})$, $\quad c_R \gets (\mc_{e + 1}, \ldots, \mc_\mn)$
        \item\label{item:recursive-call} $u^L \gets \recoverprimal(c_L, k_L')$, $\quad u^R \gets \recoverprimal(c_R, k_R')$
        \item\label{item:obtain-star} $\ustar_i =
            \begin{cases}
                u^L_i & \text{if } i < s \\ 
                1 & \text{if } i \in [s, e] \text{ and } r = i \\
                u^R_{i - e} & \text{if } i > e \\
                0 & \text{otherwise}
            \end{cases}
            $
        \item Return $\ustar$
	\end{enumerate}
\end{algorithm}



We prove the following result for our algorithm.
\begin{restatable}{theorem}{theoremrecoverprimal}
\label{theorem:recover-primal}
	The algorithm $\recoverprimal(\mc, \mk)$ solves the model projection problem for separated sparsity in time $O\left(\rtoptimalvalue{n}{(\cmax + 1)}{k}\log{n} \log{\Delta} \right)$.
\end{restatable}
We start by showing some technical lemmas. Recall that by Lemma~\ref{lemma:ystar-upper-bound} we provided an upper bound on any $\wstar_0$ such that $\wstar$ is an optimum of $\LPdeltadual$. On the other hand, there is no lower bound on $\wstar_0$. To see that, consider a very simple example $c = (1)$, $\Delta = k = 1$. Nevertheless, we can provide a lower bound in the following form.
\begin{lemma}\label{lemma:ystar-lower-bound}
    There exists an optimal solution $\wstar$ to $\LPdeltadual(\mc, \mk)$ such that $\wstar_0 \ge - (\mk - 1) \mcmax$.
\end{lemma}
\begin{proof}
To prove the lemma, we utilize the optimality condition described by Lemma~\ref{lemma:active-sets-optimality} and the following claim.
\begin{lemma}\label{lemma:active-constraints-monotone}
    Let $active^1$ be the output of $\activeconstraints(\mc, \dualgreedy(\mc, v))$ and $active^2$ be the output of $\activeconstraints(\mc, \dualgreedy(\mc, v - t))$, for any integer $v$ and a positive integer $t$. Then, it holds
    \[
        |active^1| \le |active^2|.
    \]
\end{lemma}
\begin{proof}
    Let $S$ and $S'$ be tight constraints for $\dualgreedy(\mc, v)$ and $\dualgreedy(\mc, v - t)$, respectively. Then, as we have discussed, $S \subseteq S'$. So, all we have to show is that there are at least as many active constraints formed from $S'$ as there are formed from $S$. We do that by induction on the size of $S$ under the assumption that $S \subseteq S'$.
    
    \paragraph{Base of induction: $|S| = 0$.} As the number of active constraints is non-negative, and if $|S| = 0$ there is no active constraint, the claim follows.
    
    \paragraph{Inductive step: $|S| > \mn$ and $S \subseteq S'$, for $\mn \ge 0$.} Let $i_{min}$ and $i_{min}'$ be the smallest index of $S$ and $S'$, respectively. Constraint $i_{min}$ is active for $\dualgreedy(\mc, v)$, and $i_{min}'$ is active for $\dualgreedy(\mc, v - t)$. Define $T = S \setminus \{i_{min}, \ldots, i_{min} + \Delta - 1\}$ and $T' = S' \setminus \{i_{min}', \ldots, i_{min}' + \Delta - 1\}$. Since $i_{min}' \le i_{min}$ it holds $T \subseteq T'$. Now, as for every active constraint $j$ there is no other active one in the neighborhood of $\Delta - 1$ around $j$ and except that neighborhood the other constraints are not affected by $j$, we have that
    \begin{eqnarray*}
        active^1 & = & \{i_{min}\} \cup \{\text{active constraints for } T\}, \text{ and } \\
        active^2 & = & \{i_{min'}\} \cup \{\text{active constraints for } T'\}.
    \end{eqnarray*}
    Now, as $T \subseteq T'$ and $|T| < |S|$, by inductive hypothesis we have
    \[
        |\{\text{active constraints for } T\}| \le |\{\text{active constraints for } T'\}|,
    \]
    and hence the lemma follows.
\end{proof}

    Let $\hw \gets \dualgreedy(\mc, -(k - 1) \cmax)$ and $active \gets \activeconstraints(\mc, \hw)$. We show that $|active| \ge \mk$. Furthermore, we show that $l \Delta + 1 \in active$, for all $l = 0, \ldots, \mk - 1$. Once it is shown, the claim follows by Lemma~\ref{lemma:active-constraints-monotone} and Lemma~\ref{lemma:active-sets-optimality}. Precisely, Lemma~\ref{lemma:active-constraints-monotone} shows that the number of active constraints for solutions corresponding to $\dualgreedy(\mc, -(\mk - 1) \mcmax - t)$, for $t \ge 0$, is at least $\mk$. But then, from Lemma~\ref{lemma:active-sets-optimality} we have that $\dualgreedy(\mc, -(\mk - 1) \mcmax)$ or $\dualgreedy(\mc, -(\mk - 1) \mcmax + t)$, for some $t > 0$, outputs an optimal solution.
    
    So, it only remain to show $|active| \ge k$. We prove that by induction, showing the following property. Let $i = l \Delta + 1$ be an index for some integer $l$ such that $0 \le l \le \mk - 1$. If $\hw_0 = -(\mk - 1) \mcmax$, $\sum_{j = \max\{1, i - \Delta + 1\}}^{i - 1} \hw_j \le l \mcmax$ and no constraint in $\{(l - 1) \Delta + 2, \ldots, l \Delta\}$ is active, then constraints $l \Delta + 1, (l + 1) \Delta + 1, \ldots, (\mk - 1) \Delta + 1$ are active. The induction is applied in a downward fashion on $l$, i.e. the base case is $l = \mk - 1$, and our goal is to show it holds for $l = 0$.
    
    \paragraph{Base of induction: $l = \mk - 1$.} Let $i = (\mk - 1) \Delta + 1$. As $\sum_{j = \max\{1, i - \Delta + 1\}}^{i - 1} \hw_j \le (\mk - 1) \mcmax$ and $\mc_i \ge 0$, we have $\sum_{j = \max\{1, i - \Delta + 1\}}^{i - 1} \hw_j \le \hw_0 + \mc_i$, and hence constraint $i$ is tight. Furthermore, as no constraint in $\{i - \Delta + 1, \ldots, i - 1\}$ is active, constraint $i$ is an active one.
    
    \paragraph{Inductive step: $0 \le l < \mk - 1$.} Let $i = l \Delta + 1$. First, if $\sum_{j = \max\{1, i - \Delta + 1\}}^{i - 1} \hw_j \le l \mcmax$ we have that constraint $i$ is tight, and as before active as well, and also we have $\hw_i \ge -\hw_0 - l \mcmax$.
    
    Next, we want to show that $\sum_{j = i + 1}^{i + \Delta} \hw_j \le (l + 1) \mcmax$ (note that $i + \Delta \le \mn$), so that we can use the inductive hypothesis for $l + 1$. First, we show that for every index $i'$ it holds
    \begin{equation}\label{eq:delta-sum-bound}
        \sum_{j = \max\{1, i' - \Delta + 1\}}^{i'} \le \mcmax - \hw_0.
    \end{equation}
    Recall that $\hw_0$ is negative. Towards a contradiction, assume that there exists index $q$ such that $\sum_{j = \max\{1, q - \Delta + 1\}}^{q} > \mcmax - \hw_0$. In case of tie, let $q$ be the smallest such index, which implies $\hw_q > 0$. But then, it contradicts the greedy choice of $\dualgreedy$ algorithm, as by the greedy choice we have $\hw_q = 0$ or $\sum_{j = \max\{1, q - \Delta + 1\}}^{q} = \mc_q - \hw_0 \le \mcmax - \hw_0$.
    
    Now we combine \eqref{eq:delta-sum-bound}, for $i' = i + \Delta - 1$ and $\hw_i \ge -\hw_0 - l \mcmax$ to obtain
    \[
        \sum_{j = i + 1}^{i + \Delta} \hw_j \le \mcmax - \hw_0 - \hw_i \le (l + 1) \mcmax,
    \]
    as desired. In addition, as constraint $i$ is active, we have that no constraints in $\{i + 1, \ldots, i + \Delta - 1\}$ is active, and hence we can use the inductive hypothesis.
    
    This concludes the proof.
\end{proof}

\subsection{Proof of Theorem~\ref{theorem:recover-primal}}

The algorithm $\recoverprimal(\mc, \mk)$ utilizes $\deltarecovery$ and $\distributesparsity$ to split the input problem into two subproblems and then recurses on them.
Let us now analyze the correctness and the running time of this algorithm.

\paragraph{Base case.} If the input sparsity is $0$, the algorithm outputs a zero-vector at line~\ref{line:k=0}.

\paragraph{$\Delta$-middle entries.} Next, the algorithm invokes $\deltarecovery$ over the $\Delta$ middle entries. Depending on $r$, it sets $s$ and $e$ to correspond to the interval of $\mc$ that should be removed from consideration in the recursive calls. In the same time, that intervals serves to break the input problem into two independent subproblems. It also sets $k''$ that represents the sparsity distributed outside interval $[s, e]$. Correctness of this call is guaranteed by Theorem~\ref{theorem:delta-recovery}. Value $shift$, and in turn vector $c''$, is set so that we have a guarantee that no element from the interval $[s, e]$ is chosen as part of an optimal solution. To see that, consider vector $c^{(3)} \gets c'' - shift \cdot \allones$. Then, $c_i^{(3)} = \mc_i$ for $i \notin [s, e]$ and $c_i^{(3)} = - (\mk \mcmax + 1)$. In other words, as long as $k'' \le \mk$ entries can be chosen outside of the interval $[s, e]$, no optimal solution should choose any entry within the interval $[s, e]$.

\paragraph{Distributing the sparsity.} The correctness of line~\ref{item:distribute-sparsity} follows by Theorem~\ref{theorem:distribute-sparsity}.

\paragraph{Running time.} Line~\ref{item:get-r} of algorithm $\recoverprimal$ runs in time $O\left((\log{\Delta}) \rtoptimalvalue{\mn}{(\mcmax + 1)}{\mk}\right)$ which is a subset of $O\left((\log{\Delta}) \rtoptimalvalue{\mn}{(\cmax + 1)}{k}\right)$. Vector $\mc$ is split into $c_L$ and $c_R$ at line~\ref{item:split-c} in $O(\mn)$ time. Afterwards, the method recurses on the two subproblems, and combines their outputs into $\ustar$. Obtaining $\ustar$ at line~\ref{item:obtain-star} also takes $O(\mn)$. So, it only to remain to discuss the recursion.

From Theorem~\ref{theorem:distribute-sparsity} we have that line~\ref{item:distribute-sparsity} take $O\left(\mn (\log{\mcmax} + \log{k''}) \right)$. Furthermore, as $k'' \le k$ and $\mcmax \le \cmax$, we have $O\left(\mn (\log{\mcmax} + \log{\mk}) \right) = O\left(\mn (\log{\cmax} + \log{k}) \right)$.

Every recursive step shrinks the corresponding $\mc$ vector by half at least. So, the recursion has depth of $O(\log{d})$. At every level of recursion are considered vectors $\mc$ of total length $O(d)$ -- this follows from the recursive call at line~\ref{item:recursive-call} and the fact the $c_L$ and $c_R$ represent disjoint pieces of $\mc$. So, this implies the total running time of $O\left(\log{d} (d \left(\log{\cmax} + \log{k}\right) + (\log{\Delta}) \rtoptimalvalue{d}{(\cmax + 1)}{k}) \right)$, which in a more compact way can be written as $O\left(\rtoptimalvalue{d}{(\cmax + 1)}{k} \log{d} \log{\Delta} \right)$.

\end{document}